\documentclass[11pt]{article}
\usepackage[utf8]{inputenc}
\usepackage{amssymb,amsmath,latexsym}
\usepackage[english]{babel}
\usepackage{mathptmx}
\usepackage{graphicx}
\usepackage{subfig}
\usepackage{bibentry}
\usepackage{amsthm}
\usepackage[bookmarks=true,colorlinks=true,linkcolor=blue]{hyperref}
 \usepackage[colorlinks=true]{hyperref}
\hypersetup{urlcolor=blue, citecolor=red}
\usepackage{color,transparent}

 \newtheorem{thm}{Theorem}[section]
 \newtheorem{prop}{Proposition}[section]
 \newtheorem{lem}{Lemma}[section]

 \newtheorem{dfn}{Definition}[section]

 \newtheorem{rem}{Remark}[section]
 \numberwithin{equation}{section}

\newcommand{\biindice}[3]%
{

\begin{array}[t]{c}
#1\\
{\scriptstyle #2}\\
{\scriptstyle #3}
\end{array}

}
\date{}
\begin{document}
\title{Bohr-Sommerfeld quantization conditions for Schr\"{o}dinger operator: the Method of Microlocal Wronskian and Gram Matrix}
\author{Abdelwaheb Ifa}
\maketitle
\centerline{University of Tunis El Manar, Faculty of Sciences of Tunis, Research Laboratory of PDE (LR03ES04)}
\centerline{\& University of Kairouan, Higher Institute of Applied Sciences and Technology of Kairouan, 3100 Kairouan, Tunisia}
\centerline{email: \url{abdelwaheb.ifa@fsm.rnu.tn}}
\begin{abstract}
In this work, we present a new formulation of the well known Bohr–Sommerfeld quantization rule (BS) of order 2 for a Schr\"{o}dinger operator within the algebraic and microlocal framework of B. Helffer and J. Sj\"{o}strand; BS holds precisely when the Gram matrix consisting of scalar products of some WKB solutions with respect to the "flux norm" is not invertible. This condition is obtained using the microlocal Wronskian and does not rely on traditional matching techniques. It is simplified by using action-angle variables. The interest of this procedure lies in its possible generalization to matrixvalued Hamiltonians, like Bogoliubov-de Gennes Hamiltonian.
\end{abstract}
\section{Introduction}\label{S1}
\paragraph{}
Bohr-Sommerfeld quantization rule, in its first formulation, allows to compute the energy levels $E$ of a particle in a one-dimensional potential well, its dynamics being described by the semi-classical Schr\"{o}dinger operator
\begin{equation}\label{SchrodinGERbsqrifa}
P(x,hD_{x})=(hD_{x})^{2}+V(x),\quad D_{x}=-i\,\partial_{x}.
\end{equation}
It is given at first order in $h$ by the well known formula
\begin{equation}\label{bsqrifa}
\frac{1}{2\pi h}\,\oint_{\gamma_E}\xi(x)\,dx=n+\frac{1}{2}
\end{equation}
Here $\xi(x)=\big(E-V(x)\big)^{1/2}$ denotes the momentum of the particle on its orbit $\gamma_E\subset T^*\mathbb{R}$ above the potential well $\{V(x)\leq E\}$, $n$ is an integer and the integral is computed over $\gamma_E$ in the phase space $T^{*}\mathbb{R}$.

By the implicit function theorem, we then find $E=E_n(h)$. In other words, the number of wavelengths (associated with the particle along $\gamma_E$ by de Broglie correspondence) must be an integer plus 1/2, called \emph{Maslov correction}.

This formula can be obtained in a many ways, at any order in $h$, and also admits various generalizations, for example when $P$ is the Weyl quantization of a Hamiltonian which has compact energy levels, or also for a system of such operators.
\paragraph{}
In this work we propose to set up a new method for deriving Bohr-Sommerfeld quantization rule, based on the algebraic and microlocal framework developed by B. Helffer and J. Sj\"ostrand, known as the \emph{Grushin operator}.

It consists in constructing, for a given energy $E$, local sections over $\gamma_E$ of the vector bundle $\mathcal{K}_h(E)$ of solutions $\psi_E$ to Schr\"{o}dinger equation. This bundle is endowed with a Hermitian canonical structure and a gauge group.
\paragraph{}
Let $D(E;h)$ be the determinant of the Gram matrix (for the Hermitian structure) of two branches of $\psi_E$ in a suitable basis of the co-kernel of $P$. Then Bohr-Sommerfeld quantization rule expresses the fact that $\mathcal{K}_h(E)$ is trivialized precisely when $D(E;h)=0$, so that $D(E;h)$ can be identified with the Jost function of the eigenvalue problem. Using this method, we thus recover the Bohr–Sommerfeld conditions at different orders in $h$ for Schrödinger-type operators.
\paragraph{}
We consider the semiclassical spectrum of the Schrödinger operator (\ref{SchrodinGERbsqrifa}) near the energy level $E_0<\displaystyle \liminf_{|x|\rightarrow +\infty}V(x)$. We assume that the classically allowed region is $\{V\leq E\}=[x'_E,x_E]$, where the endpoints \(x'_E\) and \(x_E\) are simple turning points, i.e. $V(x'_{E})=V(x_{E})=E$, with $V'(x'_E)<0$ and $V'(x_E)>0$.
Our aim is to derive the Bohr-Sommerfeld quantization rule for the operator \(P\), corresponding to a single potential well, by means of the microlocal Wronskian and the Gram matrix. In contrast with the traditional approach, our construction does not rely on matching techniques.
\section{The microlocal Wronskian}\label{S2}
\paragraph{}
Let $p(x,\xi;h)$ be a smooth real classical Hamiltonian on $T^*\mathbb{R}$; we will assume that $p$ belongs to the space of symbols $S^0(m)$ for some order function $m$ \big(for example $m(x,\xi)=(1+|\xi|^{2})^{M}$\big) with
$$S^{N}(m)=\big\{p\in C^{\infty}(T^{*}\mathbb{R}): \forall\,\alpha\in \mathbb{N}^{2},\,\exists\,C_{\alpha}>0,\quad |\partial_{(x,\xi)}^\alpha p(x,\xi;h)|\leq C_{\alpha}\,h^{N}\,m(x,\xi)\big\}$$
This allows to take Weyl quantization of $p$
\begin{equation}\label{A01WEyLOO}
P(x,hD_x;h)u(x; h)=(2\,\pi\,h)^{-1}\,\int\int e^{\frac{i}{h}\,(x-y)\,\eta}\,p\big(\frac{x+y}{2},\eta;h\big)\,u(y)\,dy\,d\eta.
\end{equation}
so that $P(x,hD_{x};h)$ is self-adjoint. We also assume that $p+i$ is elliptic (in the classical sense). We call as usual $p_{0}$ the principal symbol, and $p_{1}$ the sub-principal symbol; in case of Schrödinger operator (\ref{SchrodinGERbsqrifa}), $p(x,\xi;h)=p_{0}(x,\xi)=\xi^{2}+V(x)$. We make the geometric assumption (H) of~\cite{CdV1}:
\paragraph{}
Fix some compact interval $I=[E_-, E_+]$, $E_{-}<E_{+}$ and assume that there exists a topological ring $\mathcal{A}\subset p_0^{-1}(I)$ such that $\partial \mathcal{A}=A_+\cup A_-$ with $A_{\pm}$ a connected component of $p_0^{-1}(E_{\pm})$. Assume also that $p_{0}$ has no critical point in $\mathcal{A}$, and $A_{-}$ is included in the disk bounded by $A_{+}$ (if this is not the case, we can always change $p$ to $-p$). We define the microlocal well $W$ as the disk bounded by par $A_{+}$. These conditions guarantee that the spectrum of $P$ in $I$ is discrete and can be described semiclassically.
\paragraph{}
For each \(E\in I\), let \(\gamma_E \subset W\) be a periodic orbit in the energy surface \(\{ p_0(x,\xi)=E\}\), so that \(\gamma_E\) is an embedded Lagrangian manifold.
\paragraph{}
Let $\mathcal{K}_{h}^{N}(E)$ be the microlocal kernel of $P-E$ of order $N$, i.e. the space of local solutions of $(P-E)u=\mathcal{O}(h^{N+1})$ in the distributional sense, microlocalized on $\gamma_{E}$. This is a smooth complex vector bundle over $\pi_{x}(\gamma_{E})$. Here we address the problem of finding the set of $E=E(h)$ such that $\mathcal{K}_{h}^{N}(E)$ contains a global section, i.e. of constructing a sequence of quasi-modes $\big(u_{n}(h), E_{n}(h)\big)$ of a given order $N$. As usual we denote by $\mathcal{K}_{h}(E)$ the microlocal kernel of $P-E$ mod $\mathcal{O}(h^{\infty})$; since the distinction between $\mathcal{K}_{h}^{N}(E)$ and $\mathcal{K}_{h}(E)$ plays no important role here, we shall content to write $\mathcal{K}_{h}(E)$.
\paragraph{}
The best algebraic and microlocal framework for computing 1-D quantization rules in the self-adjoint case, developed in the fundamental works of \cite{JO22025}, \cite{HARPE0RIII}, is based on Fredholm theory and the classical \emph{positive commutator method}, which involves conservation of a quantity called \emph{quantum flux}.
\paragraph{}
Bohr-Sommerfeld quantization rules are derived by constructing quasi-modes using the WKB approximation along a closed Lagrangian manifold \(\Lambda_E\subset \{p_0=E\}\), i.e. a periodic orbit of the Hamiltonian vector field \(H_p\) with energy \(E\). This construction is local and depends on the rank of the projection \(\Lambda_E \rightarrow\mathbb{R}_x\).
\paragraph{}
Thus, the set \(K_h(E)\) of microlocal solutions to \((P-E)u=0\) along \(\Lambda_E\) can be seen as a bundle over $\mathbb{R}$ with a compact base, corresponding to the classically allowed region at energy \(E\). The eigenvalues \(E_n(h)\) are then determined by the condition that the global quasi-mode obtained by gluing local WKB solutions along \(\Lambda_E\) is singlevalued, i.e. that \(K_h(E)\) has trivial holonomy.
\paragraph{}
Assuming \(\Lambda_E\) is smoothly embedded in \(T^*\mathbb{R}\), it can always be parametrized by a non-degenerate phase function. Of particular interest are the \emph{focal points}, i.e. critical points of the phase functions, which are responsible for the change in Maslov index. A point \(a_{E}=(x_E, \xi_E)\in\Lambda_E\) is a focal point if \(\Lambda_E\) "turns vertical" at \(a_{E}\), meaning that the tangent space \(T_{a_{E}}\Lambda_E\) is no longer transverse to the fiber \(x=\text{const}\). in \(T^*\mathbb{R}\).
\paragraph{}
In any case however, \(\Lambda_E\) can locally be parametrized either by a phase function \(S(x)\) (spatial representation) or \(\widetilde{S}(\xi)\) (Fourier representation). We fix an orientation on \(\Lambda_E\) and for any point \(a\in\Lambda_E\) (not necessarily a focal point), we denote by \(\rho=\pm1\) the oriented segments near \(a\). Let \(\chi^{a}\in C_0^\infty(\mathbb{R}^2)\) be a smooth cut-off function equal to 1 near \(a\), and \(\omega^{a}_{\rho}\) a small neighborhood of \(\operatorname{supp}[P,\chi^{a}]\cap\Lambda_E\) near \(\rho\).
\begin{dfn}\label{A02}
Let $P$ be self-adjoint, and $u^{a}, v^{a}\in K_h(E)$ be supported on $\Lambda_E$. We call
\begin{equation}\label{A03}
\mathcal{W}^{a}_{\rho}\big(u^{a},\overline{v^{a}}\big)=\big(\frac{i}{h}\,[P,\chi^{a}]_{\rho}u^{a}|v^{a}\big)
\end{equation}
the \emph{microlocal Wronskian} of $(u^{a}, \overline{v^{a}})$ in $\omega^{a}_{\rho}$. Here $\displaystyle\frac{i}{h}[P,\chi^{a}]_{\rho}$ denotes the part of the commutator supported microlocally on $\omega^{a}_{\rho}$.
\end{dfn}
To understand that terminology, let $P(x,hD_{x})=(hD_{x})^{2}+V(x)$, $x_{E}=0$ and change $\chi$ to Heaviside unit step-function $\chi(x)$, depending on $x$ alone. Then in distributional sense, we have $\displaystyle\frac{i}{h}[P,\chi]=-i\,h\,\chi''+2\,\chi'\,hD_x=-i\,h\,\delta'+2\,\delta\,hD_x$, where $\delta$ denotes the Dirac measure at $0$, and $\delta'$ its derivative, so that $\big(\displaystyle {i\over h}[P,\chi]u|u\big)=-ih\big(u'(0)\overline{u(0)}-u(0)\overline{u'(0)}\big)$ is the usual Wronskian of $(u, \overline{u})$.
\begin{prop}
Let $u^{a}, v^{a}\in K_{h}(E)$, and denote by $\widehat{u}$ the $h$-Fourier (unitary) transform of $u$. Then mod $\mathcal{O}(h^{\infty})$:
\begin{equation}\label{A04}
\mathcal{W}=\big(\frac{i}{h}\,[P,\chi^{a}] u^{a}|v^{a}\big)=\big(\frac{i}{h}\,[P,\chi^{a}] \widehat u^{a}|\widehat v^{a}\big)=0
\end{equation}
\begin{equation}\label{A05}
\mathcal{W}^{a}_{+}\big(u^{a},\overline{v^{a}}\big)=-\mathcal{W}^{a}_{-}\big(u^{a},\overline{v^{a}}\big)
\end{equation}
Moreover, $\mathcal{W}^{a}_{\rho}\big(u^{a}, \overline{v^{a}}\big)$ does not depend on the choice of $\chi^{a}$ above.
\end{prop}
\begin{proof}
Since $u^{a}, v^{a}\in K_{h}(E)$ are distributions in $L^{2}$, the first equality (\ref{A04}) follows from the Plancherel formula and the regularity of microlocal solutions in $L^{2}$, $p+i$ being elliptic. If $a$ is not a focal point,  $u^{a}, v^{a}$ are smooth WKB solutions near $a$, so we can expand the commutator in $\mathcal{W}=\big(\frac{i}{h}\,[P,\chi^{a}] u^{a}|v^{a}\big)$ and use that $P$ is self-adjoint to show that $\mathcal{W}=\mathcal{O}(h^{\infty})$. If $a$ is a focal point, $u^{a}, v^{a}$ are smooth WKB solutions in Fourier representation, so again $\mathcal{W}=\mathcal{O}(h^{\infty})$. Then (\ref{A05}) follows from Definition \ref{A02}.
\end{proof}
\section{Semiclassical Quantization Rules, Quasi-Modes, and the Flux Norm}
For simplicity, we assume that the principal symbol $p_{0}$ of $P$ has only two focal points along $\gamma_{E}$. However, it is clear from the matching of microlocal solutions that the result does not depend on this assumption. In fact, the Bohr–Sommerfeld condition depends on the geometry of $\gamma_{E}$ only through its Maslov index, which equals 2 when $\gamma_{E}$ is an embedded Lagrangian submanifold. The case where $\gamma_{E}$ is not a submanifold (homeomorphic to a figure eight), for which the Maslov index is 0, is not considered here, see \cite{JO22025} or \cite{YCVPAR1}. (Recall that in dimension 1, the Maslov index, defined modulo 4, is always even, hence equal to 0 or 2).

Our main result is the following:
\begin{thm}
With the notations and hypotheses stated above, for $h>0$ small enough there exists a smooth function $\mathcal{S}_{h}:I\rightarrow \mathbb{R}$, called the
semi-classical action, with asymptotic expansion $\mathcal{S}_{h}(E)=\displaystyle\sum_{j=0}^{\infty}h^{j}\,S_{j}(E)$ such that $E\in I$ is an eigenvalue of $P$ iff it satisfies the implicit equation (Bohr-Sommerfeld quantization condition) $\mathcal{S}_{h}(E)=2\pi n h$,\,$n\in \mathbb{Z}$. The semi-classical action consists of:
\begin{itemize}
\item the classical action along $\gamma_{E}$
$$S_0(E)=\oint_{\gamma_E}\xi(x)\,dx=\displaystyle\int\int_{\{p_0\leq E\}\cap W}\,d\xi\wedge\,dx,$$
\item the first-order correction, including the Maslov index $S_1(E)=-\pi$,
\item the second order term
 $$S_2(E)=\frac{1}{12}\,\frac{d}{dE}\int_{\gamma_E}V''\big(x(t)\big)\,dt.$$
\end{itemize}
\end{thm}
\subsection{Normalized Quasi-Modes}
We first recall H\"{o}rmander's asymptotic stationary phase theorem (see e.g. \cite{LARsHORmander}, Theorem 7.7.5):

Let $f:\mathbb{R}^{d}\rightarrow \mathbb{C}$ be a function such that $\mathrm{Im}\big(f(x)\big)\geq 0$, and suppose that $f$ has a non-degenerate critical point at $x_{0}$. Then, we have the asymptotic expansion:
\begin{equation}\label{Phasestat0}
\int_{\mathbb{R}^{d}}e^{\frac{i}{h}\,f(x)}\,u(x)\,dx \sim e^{\frac{i}{h}\,\varphi(x_{0})}\,\bigg(\det\big(\frac{\varphi''(x_{0})}{2\,i\,\pi\,h}\big)\bigg)^{-\frac{1}{2}}\,\sum_{j}h^{j}\,L_{j} u (x_{0})
\end{equation}
where $L_{j}$ are linear differential operators, with $L_{0}u(x_{0})=u(x_{0})$ and in particular:
\begin{equation}\label{TTL551A06}
L_{1}u(x_{0})=\sum_{n=0}^{2}\frac{2^{-(n+1)}}{i\,n!\,(n+1)!}\,\big<\big(f''(x_{0})\big)^{-1}\,D_{x},
D_{x}\big>^{n+1}\,(\Phi_{x_{0}}^{n}\,u)(x_{0})
\end{equation}
with:
\begin{equation}\label{A07}
\Phi_{x_{0}}(x)=f(x)-f(x_{0})-\frac{1}{2}\,\big<f''(x_{0}).\,(x-x_{0}), x-x_{0}\big>
\end{equation}
We note that $\Phi_{x_{0}}$ vanishes to order 3 at $x_{0}$ \big(i.e., $\Phi_{x_{0}}(x_{0})=0,\;\Phi'_{x_{0}}(x_{0})=0,\;\Phi''_{x_{0}}(x_{0})=0$\big).
\subsubsection{Asymptotic solutions in the Fourier representation}
Near the focal point $a_E=(x_E ,0)$, the energy surface $C_E=\{\,\xi^2+V(x)=E\,\}$ is the graph of a function \(x=-\psi'(\xi)\). Let \(\xi \in J\), where \(J=[-A,A]\) is an interval containing \(0\), and assume \(x_E>x_1=-\psi'(A)\). It is therefore natural to try a WKB form for \(\mathcal{F}_h u=\widehat{u}\) (the semiclassical Fourier transform of \(u\)) in \(\widetilde{I}=[x_1, x_E + 1]\), namely
\begin{equation}\label{MASLOV2}
u(x)=(2\,\pi\,h)^{-\frac{1}{2}}\,\int e^{\frac{i}{h}\,\big(x\,\xi+\psi(\xi)\big)}\,b(\xi;h)\,d\xi
\end{equation}
with \(b = b(\cdot;h) \in C_0^\infty(\mathbb{R})\) a symbol in \(h\). When \(x\) is near the turning point $x_{E}$, we must recover the WKB Ansatz, or more precisely, a linear combination of \(u_+\) and \(u_-\). A straightforward computation shows that
\begin{equation}\label{MASLOV3}
\big(P(x,hD_{x})-E\big)u(x;h)=(2\,\pi\,h)^{-\frac{1}{2}}\,\int e^{\frac{i}{h}\,\big(x\,\xi+\psi(\xi)\big)}\,\big(p_{0}(x,\xi)-E\big)\,b(\xi;h)\,d\xi
\end{equation}
where $p_0(x,\xi)=\xi^2+V(x)$ denotes the principal symbol of the Schr\"{o}dinger operator $P$.

We now look for $b(\xi;h)\sim b_{0}(\xi)+h\,b_{1}(\xi)+h^{2}\,b_{2}(\xi)+\cdots$ a classical elliptic symbol with \(b_0(0)\neq 0\), such that there exists $a(x,\xi;h)\sim a_{0}(x,\xi)+h\,a_{1}(x,\xi)+h^{2}\,a_{2}(x,\xi)+\cdots$ satisfying
\begin{equation}\label{E00}
e^{\frac{i}{h}\,\big(x\,\xi+\psi(\xi)\big)}\,\big(p_{0}(x,\xi)-E\big)\,b(\xi;h)=
h\,D_{\xi}\Big[e^{\frac{i}{h}\,\big(x\,\xi+\psi(\xi)\big)}\,a(x,\xi;h)\Big]
\end{equation}
Therefore, for $\xi$ near $0$, $x$ near $x_E$, and $E$ near $E_0$, it suffices to solve the equations
\begin{equation}\label{E1}
\big(p_{0}(x,\xi)-E\big)\,b_{0}(\xi)=\big(x+\psi'(\xi)\big)\,a_{0}(x,\xi)
\end{equation}
\begin{equation}\label{E2}
\big(p_{0}(x,\xi)-E\big)\,
b_{1}(\xi)=\big(x+\psi'(\xi)\big)\,a_{1}(x,\xi)+D_{\xi}\,a_{0}(x,\xi)
\end{equation}
\begin{equation}\label{E3}
\big(p_{0}(x,\xi)-E\big)\,
b_{2}(\xi)=\big(x+\psi'(\xi)\big)\,a_{2}(x,\xi)+D_{\xi}\,a_{1}(x,\xi)
\end{equation}
Here, we have collected the terms according to the powers of \(h\) in (\ref{E00}); equation (\ref{E1}) is obtained by annihilating the term of order \(h^0\) in (\ref{E00}), equation (\ref{E2}) is obtained by annihilating the term of order \(h^1\), and equation (\ref{E3}) by annihilating the term of order \(h^2\).

Define $\lambda(x,\xi)$ by
\begin{equation}\label{029E4}
\lambda(x,\xi):=\frac{p_{0}(x,\xi)-E}{x+\psi'(\xi)}
\end{equation}
From (\ref{029E4}), we deduce that
\begin{equation}\label{E5}
\lambda\big(-\psi'(\xi),\xi\big)=\big(\partial_{x}\,p_{0}\big)\big(-\psi'(\xi),\xi\big)=V'\big(-\psi'(\xi)\big):=\alpha(\xi)
\end{equation}
Differentiating with respect to \(\xi\) both sides of the Hamilton–Jacobi equation
\begin{equation}\label{Jacob}
p_{0}\big(-\psi'(\xi),\xi\big)=E
\end{equation}
we obtain
\begin{equation}\label{E6}
\psi''(\xi)=\frac{\big(\partial_{\xi}\,p_{0}\big)\big(-\psi'(\xi),\xi\big)}{\alpha(\xi)}=\frac{2\,\xi}{\alpha(\xi)}
\end{equation}
which vanishes at \(\xi_E=0\). For a given \(b_0\), the unique solution of (\ref{E1}) is
\[a_{0}(x,\xi)=\lambda(x,\xi)\,b_{0}(\xi).\]
To solve (\ref{E2}), it is necessary and sufficient that
\begin{equation}\label{E888}
\big(D_{\xi}\,a_{0}\big)\big(-\psi'(\xi),\xi\big)=0
\end{equation}
Equation (\ref{E888}) is equivalent to
\begin{equation}\label{E6669}
\big(\partial_{\xi}\,\lambda\big)\big(-\psi'(\xi),\xi\big)\,b_{0}(\xi)+\alpha(\xi)\,b'_{0}(\xi)=0
\end{equation}
From (\ref{029E4}), one finds
\begin{equation}\label{RelaDx1999DxLamda}
\big(\partial_{\xi}\,\lambda\big)\big(-\psi'(\xi),\xi\big)=-\psi''(\xi)\,\big(\partial_{x}\,\lambda\big)\big(-\psi'(\xi),\xi\big)
\end{equation}
A direct computation using (\ref{029E4}) shows that
\begin{equation}\label{E8}
\big(\partial_{x}\,\lambda\big)\big(-\psi'(\xi),\xi\big)=
\frac{1}{2}\,\big(\frac{\partial^{2} p_{0}}{\partial x^{2}}\big)\big(-\psi'(\xi),\xi\big)=\frac{1}{2}\,V''\big(-\psi'(\xi)\big)
\end{equation}
Therefore,
\begin{equation}\label{}
\big(\partial_{\xi}\,\lambda\big)\big(-\psi'(\xi),\xi\big)=-\frac{1}{2}\,\psi''(\xi)\,V''\big(-\psi'(\xi)\big)=\frac{1}{2}\,\alpha'(\xi)
\end{equation}
Thus, equation (\ref{E6669}) is equivalent to
\begin{equation}\label{EQDIBOOO}
\frac{1}{2}\,\alpha'(\xi)\,b_{0}(\xi)+\alpha(\xi)\,b'_{0}(\xi)=0
\end{equation}
As the projection onto $\mathbb{R}_{\xi}$ of the classical dynamics is given by $-\displaystyle\alpha\,\frac{d}{d\xi}$, the transport equation (\ref{EQDIBOOO}) can be interpreted as the invariance of $b_{0}^{2}\,|d\xi|$ under the dynamics.

The differential equation (\ref{EQDIBOOO}) is solvable by quadrature, and its solutions are functions of the form
\begin{equation}\label{QUADBZEROEQDIBOOO}
b_{0}(\xi)=C_{0}\,|\alpha(\xi)|^{-1/2}
\end{equation}
To solve equation (\ref{E3}), it is necessary and sufficient that
\begin{equation}\label{E13}
\big(D_{\xi}\,a_{1}\big)\big(-\psi'(\xi),\xi\big)=0
\end{equation}
Equation (\ref{E2}) yields
\begin{equation}\label{E14}
a_{1}(x,\xi)=\lambda(x,\xi)\,b_{1}(\xi)+\lambda_{0}(x,\xi)
\end{equation}
where we have set
\begin{equation}\label{ELamdaZEroo15}
\lambda_{0}(x,\xi)=-\frac{D_{\xi}\,a_{0}(x,\xi)}{x+\psi'(\xi)}=\frac{i\,\partial_{\xi}\,a_{0}(x,\xi)}{x+\psi'(\xi)}
\end{equation}
Equation (\ref{E13}) is equivalent to
\begin{equation}\label{T00RANSPORTB636}
\alpha(\xi)\,b'_{1}(\xi)+\frac{1}{2}\,\alpha'(\xi)\,b_{1}(\xi)=-(\partial_{\xi}\lambda_{0})\big(-\psi'(\xi),\xi\big)
\end{equation}
A simple computation shows that
\begin{equation}\label{E22}
\lambda_{0}\big(-\psi'(\xi),\xi\big)=-
i\,b_{0}(\xi)\,\big[\frac{\psi''(\xi)}{6}\frac{\partial^{3} p_{0}}{\partial x^{3}}
+\frac{\alpha'}{4\,\alpha}\,\frac{\partial^{2} p_{0}}{\partial x^{2}}\big]_{x=-\psi'(\xi)}
\end{equation}
\begin{equation}\label{E23}
\big(\partial_{x}\lambda_{0}\big)\big(-\psi'(\xi),\xi\big)=-
i\,\frac{b_{0}(\xi)}{12}\,\Big(
\frac{\psi''(\xi)}{2}\,\frac{\partial^{4} p_{0}}{\partial x^{4}}+
\frac{\alpha'(\xi)}{\alpha(\xi)}\,\frac{\partial^{3} p_{0}}{\partial x^{3}}\Big)_{x=-\psi'(\xi)}
\end{equation}
Differentiating both sides of equation (\ref{E22}) with respect to $\xi$, we obtain
\begin{align}
\big[-\psi''(\xi)\,\partial_{x}\lambda_{0}+\partial_{\xi}\lambda_{0}\big]_{x=-\psi'(\xi)}
=i\,b_{0}(\xi)\,\Big[\frac{(\psi'')^{2}}{6}\,\frac{\partial^{4}p_{0}}{\partial x^{4}}+\big(\frac{2\theta_{0}\,\psi''}{3}-\frac{\psi'''}{6}\big)\,
\frac{\partial^{3}p_{0}}{\partial x^{3}}+\frac{1}{2}\,(\theta_{0}^{2}-\theta'_{0})\,\frac{\partial^{2}p_{0}}{\partial x^{2}}\Big]_{x=-\psi'(\xi)}
\end{align}
From formula (\ref{E23}), it follows that
\begin{equation}\label{Dxilamda7ZErocritnnnn}
\begin{aligned}
(\partial_{\xi}\lambda_{0})\big(-\psi'(\xi),\xi\big)&=i\,b_{0}(\xi)\,
\Big[\frac{(\psi'')^{2}}{8}\,\frac{\partial^{4}p_{0}}{\partial x^{4}}+\big(\frac{\theta_{0}\,\psi''}{2}-\frac{\psi'''}{6}\big)\,
\frac{\partial^{3}p_{0}}{\partial x^{3}}+\frac{1}{2}\,(\theta_{0}^{2}-\theta'_{0})\,\frac{\partial^{2}p_{0}}{\partial x^{2}}\Big]_{x=-\psi'(\xi)}\ \\
&=i\,C_{0}\,|\alpha(\xi)|^{-1/2}\,\Big[\frac{(\psi'')^{2}}{8}\,\frac{\partial^{4}p_{0}}{\partial x^{4}}+\big(\frac{\theta_{0}\,\psi''}{2}-\frac{\psi'''}{6}\big)\,
\frac{\partial^{3}p_{0}}{\partial x^{3}}+\frac{1}{2}\,(\theta_{0}^{2}-\theta'_{0})\,\frac{\partial^{2}p_{0}}{\partial x^{2}}\Big]_{x=-\psi'(\xi)}
\end{aligned}
\end{equation}
where we have defined
\begin{equation}\label{Theta111indzer0}
\theta_{0}=\frac{\alpha'}{2\,\alpha}
\end{equation}
We also used the relations
\begin{equation}\label{GRANDBZERO}
b'_{0}=-\theta_{0}\,b_{0},\quad b''_{0}=\big(\theta_{0}^{2}-\theta'_{0}\big)\,b_{0}
\end{equation}
The homogeneous equation associated with (\ref{T00RANSPORTB636}) is the same as (\ref{EQDIBOOO}). Therefore, we seek a particular solution of (\ref{T00RANSPORTB636}) of the form
\begin{equation}\label{E18}
D_{1}(\xi)\,|\alpha(\xi)|^{-\frac{1}{2}}
\end{equation}
By the method of variation of constants, one finds
\begin{equation}\label{0E019}
D_{1}(\xi)=-i\,C_{0}\,\int_{0}^{\xi}\frac{1}{\alpha}\,
\Big[\frac{(\psi'')^{2}}{8}\,V^{(4)}+\big(\frac{\theta_{0}\,\psi''}{2}-\frac{\psi'''}{6}\big)\,
V'''+\frac{1}{2}\,(\theta_{0}^{2}-\theta'_{0})\,V''\Big]\,d\zeta
\end{equation}
By integration by parts, we obtain
\begin{equation}\label{INTEGPARt0E0ZZ19}
\frac{1}{C_{0}}\,\text{Im}\big(D_{1}(\xi)\big)
=\int_{0}^{\xi}\frac{1}{\alpha}\,
\Big[\frac{(\psi'')^{2}}{24}\,V^{(4)}+\frac{\alpha'\,\psi''}{6\,\alpha}\,
V'''+\frac{(\alpha')^{2}}{8\,\alpha^{2}}\,V''\Big]_{x=-\psi'(\zeta)}d\zeta+
\big[\frac{\psi''}{6\,\alpha}\,V'''+\frac{\alpha'}{4\,\alpha^{2}}\,V''\big]_{0}^{\xi}
\end{equation}
We normalize by imposing
\begin{equation}\label{E20}
D_{1}(0)=0
\end{equation}
so that the general solution of (\ref{T00RANSPORTB636}) takes the form
\begin{equation}\label{E21}
b_{1}(\xi)=\big(C_{1}+D_{1}(\xi)\big)
\,|\alpha(\xi)|^{-\frac{1}{2}}
\end{equation}
It then follows that
\begin{equation}\label{}
b_{0}(\xi)+h\,b_{1}(\xi)=\big(C_{0}+h\,C_{1}+h\,D_{1}(\xi)\big)
\,|\alpha(\xi)|^{-\frac{1}{2}}
\end{equation}
The integration constants $C_0$ and $C_1$ will be determined by normalizing the microlocal Wronskians as follows:
\subsubsection{Normalisation}\label{nOrMaLiSaTiOnn}
We compute the microlocal Wronskian of $(u^{a},\overline{u^{a}})=(u,\overline{u})$ in $\omega^{a}_{\rho}$. Our goal is to normalize the microlocal solution
\begin{equation}\label{MICROLOCALsolutionFOURIER}
\widehat{u}(\xi;h):=e^{\frac{i}{h} \psi(\xi)}\,b(\xi;h)=e^{\frac{i}{h} \psi(\xi)}\,\big(C_{0}+h\,C_{1}+h\,D_{1}(\xi)+O(h^{2})\big)\,|\alpha(\xi)|^{-\frac{1}{2}}
\end{equation}
using the microlocal Wronskian. That is, we seek constants \(C_0\) and \(C_1\) such that
\[\mathcal{W}^{a}\big(\widehat{u},\overline{\widehat{u}}\big)=1+\mathcal{O}(h^{2}).\]
In the Fourier representation, we have
$$\mathcal{W}^{a}_{\rho}\big(\widehat{u},\overline{\widehat{u}}\big)=\big(\frac{i}{h}\,[P,\chi^{a}]_{\rho}\,\widehat{u}|\widehat{u}\big),$$
where $\chi^{a}\in C^{\infty}_{0}(\mathbb{R}^{2})$ is a smooth cut-off equal to 1 near the focal point $a=(x_{E},0)$. Without loss of generality, we can take $\chi^{a}(x,\xi)=\chi_{1}^{a}(x)\,\chi_{2}(\xi)$, with $\chi_{2}\equiv 1$ on small neighborhoods $\omega^{a}_{\pm}$, of $\mathrm{supp}\big(\frac{i}{h}\,[P,\chi^{a}]\big)\cap \{\xi^{2}+V(x)=E\}$ in the region $\pm \xi>0$. Thus it is sufficient to consider variations of $\chi_1^{a}$ only.

In general, if $P$ and $Q$ are two $h$-pseudodifferential operators, whose Weyl symbols are respectively given by
\begin{equation*}
\sigma^{W}(P)(x,\xi;h)=p_{0}(x,\xi)+h\,p_{1}(x,\xi)+h^{2}\,p_{2}(x,\xi)+\cdots
\end{equation*}
\begin{equation*}
\sigma^{W}(Q)(x,\xi;h)=q_{0}(x,\xi)+h\,q_{1}(x,\xi)+h^{2}\,q_{2}(x,\xi)+\cdots
\end{equation*}
then
\begin{equation*}
\sigma^{W}\big(\frac{i}{h}[P,Q]\big)=\{\sigma^{W}(P),\sigma^{W}(Q)\}+\mathcal{O}(h^{2})
\end{equation*}
where \(\{\cdot,\cdot\}\) denotes the Poisson bracket. In particular,
\begin{equation*}
\sigma^{W}\big(\frac{i}{h}\,[P,\chi^{a}]\big)(x,\xi;h)=2\,\xi\,\chi'_{1}(x)+\mathcal{O}(h^{2}):=c(x,\xi;h)+\mathcal{O}(h^{2})
\end{equation*}
Using the Weyl calculus, the operator acts in Fourier representation as
\begin{equation*}
c^{w}(-h\,D_{\xi},\xi;h)\,v(\xi;h)=(2\,\pi\,h)^{-1}\,\int\int e^{-\frac{i}{h}\,(\xi-\eta)\,y}\,c\big(y,\frac{\xi+\eta}{2};h\big)\,v(\eta;h)\,dy\,d\eta
\end{equation*}
Hence, applying this to $\widehat{u}$, we obtain
\begin{equation*}
\frac{i}{h}\,[P,\chi^{a}]\,\widehat{u}(\xi;h)=
(2\,\pi\,h)^{-1}\,\int\int e^{\frac{i}{h}\big(\psi(\eta)-(\xi-\eta)\,y\big)}\,c\big(y,\frac{\xi+\eta}{2};h\big)\,b(\eta;h)\,dy\,d\eta
\end{equation*}
For fixed $\xi$, the phase corresponding to the oscillatory integral defining $\displaystyle \frac{i}{h}\,[P,\chi^{a}]\,\widehat{u}$ is
\begin{equation*}
\varphi_{\xi}(y,\eta)=\psi(\eta)-(\xi-\eta)\,y
\end{equation*}
whose critical points are
\begin{equation*}
\big(y_{c}(\xi),\eta_{c}(\xi)\big)=\big(-\psi'(\xi),\xi\big)
\end{equation*}
and therefore, the critical values of $\varphi_{\xi}$ are
\begin{equation*}
\varphi_{\xi}\big(y_{c}(\xi),\eta_{c}(\xi)\big)=\psi(\xi)
\end{equation*}
A direct computation shows that
\begin{equation*}
\big(\mathrm{Hess}\,\varphi_{\xi}\big)\big(y_{c}(\xi),\eta_{c}(\xi)\big)=
\left(
\begin{array}{cc}
0 & 1 \\
1 & \psi''(\xi) \\
\end{array}
\right)
\end{equation*}
By the stationary phase theorem already stated in (\ref{Phasestat0}), we obtain:
\begin{equation*}
\frac{i}{h}\,[P,\chi^{a}]\,\hat{u}(\xi;h)\sim e^{\frac{i}{h}\,\psi(\xi)}\,\sum_{j}h^{j}\,L_{j}v_{\xi}\big(-\psi'(\xi),\xi;h\big)
\end{equation*}
where we have set
\begin{equation*}
v_{\xi}(y,\eta;h)=c\big(y,\frac{\xi+\eta}{2};h\big)\,b(\eta;h)=(\xi+\eta)\,\big(\chi_{1}^{a}\big)'(y)\,b(\eta;h)
\end{equation*}
In particular,
\begin{equation}\label{LZEROVcritique}
L_{0}v_{\xi}\big(-\psi'(\xi),\xi;h\big)=v_{\xi}\big(-\psi'(\xi),\xi;h\big)=c\big(-\psi'(\xi),\xi;h\big)\,b(\xi;h)
=2\,\xi\,\big(\chi_{1}^{a}\big)'\big(-\psi'(\xi)\big)\,b(\xi;h)
\end{equation}
A direct computation using (\ref{TTL551A06}) shows that
\begin{align*}
L_{1}v_{\xi}\big(-\psi'(\xi),\xi;h\big)&
=\sum_{n=0}^{2}\frac{2^{-(n+1)}}{i\,n!\,(n+1)!}\,\big<\left(
\begin{array}{cc}
-\psi''(\xi) & 1 \\
1 & 0 \\
\end{array}
\right)    \,\left(
               \begin{array}{c}
                 D_{y} \\
                 D_{\eta} \\
               \end{array}
             \right)
, \left(
               \begin{array}{c}
                 D_{y} \\
                 D_{\eta} \\
               \end{array}
             \right)\big>^{n+1}\,(g^{n}_{\xi}\,v_{\xi})\big(-\psi'(\xi),\xi;h\big)\ \\
             &=\sum_{n=0}^{2}\frac{2^{-(n+1)}}{i\,n!\,(n+1)!}\,
\Big(\psi''(\xi)\,\frac{\partial^{2}}{\partial y^{2}}-2\,\frac{\partial^{2}}{\partial y \partial \eta}\Big)^{n+1}\,
(g^{n}_{\xi}\,v_{\xi})\big(-\psi'(\xi),\xi;h\big)
\end{align*}
A straightforward computation shows that
\begin{align*}
g_{\xi}(y,\eta)&=\varphi_{\xi}(y,\eta)-\varphi_{\xi}\big(-\psi'(\xi),\xi\big)-\frac{1}{2}\,\big<\left(
\begin{array}{cc}
 0 & 1 \\
  1 & \psi''(\xi) \\
  \end{array}
  \right)\,\left(
   \begin{array}{c}
      y+\psi'(\xi) \\
       \eta-\xi\\
      \end{array}
       \right)
       ,\left(
             \begin{array}{c}
                     y+\psi'(\xi) \\
                              \eta-\xi\\
                                     \end{array}
                                   \right)
\big>\ \\
&=\psi(\eta)-\psi(\xi)-(\eta-\xi)\,\psi'(\xi)-\frac{1}{2}\,(\eta-\xi)^{2}\,\psi''(\xi)=\mathcal{O}\big((\eta-\xi)^{3}\big)
\end{align*}
Notice that $g_{\xi}(y,\eta)$ does not depend on the variable $y$. After lengthy calculations, one shows that
\begin{equation*}
\Big(\psi''(\xi)\,\frac{\partial^{2}}{\partial y^{2}}-2\,\frac{\partial^{2}}{\partial y \partial \eta}\Big)^{n+1}\,
(g^{n}_{\xi}\,v_{\xi})\big(-\psi'(\xi),\xi;h\big)=0,\quad\forall\,n=1, 2
\end{equation*}
and therefore
\begin{equation*}
L_{1}\,v\big(-\psi'(\xi),\xi;h\big)=\frac{1}{2\,i}\,\Big(\psi''(\xi)\,\frac{\partial^{2}v_{\xi}}{\partial y^{2}}-2\,\frac{\partial^{2}v_{\xi}}{\partial y \partial \eta}\Big)\big(-\psi'(\xi),\xi;h\big)
=i\,b(\xi;h)\,r'(\xi)+2\,i\,r(\xi)\,\partial_{\xi}b(\xi;h)
\end{equation*}
with
\begin{equation}\label{FormuLERxi}
r(\xi)=\xi\,\big(\chi_{1}^{a}\big)''\big(-\psi'(\xi)\big)=\frac{\alpha(\xi)\,\psi''(\xi)}{2}\,\big(\chi_{1}^{a}\big)''\big(-\psi'(\xi)\big)
\end{equation}
It follows that
\begin{equation}\label{}
\frac{i}{h}\,[P,\chi^{a}]\,\widehat{u}(\xi;h)\sim e^{\frac{i}{h}\,\psi(\xi)}\,d(\xi;h)
\end{equation}
where we have set
\begin{equation*}
d(\xi;h)=d_{0}(\xi)+h\,d_{1}(\xi)+\mathcal{O}(h^{2})
\end{equation*}
with
\begin{equation}\label{}
d_{0}(\xi)=2\,\xi\,\big(\chi_{1}^{a}\big)'\big(-\psi'(\xi)\big)\,b_{0}(\xi)
\end{equation}
and
\begin{equation}\label{}
d_{1}(\xi)=i\,r'(\xi)\,b_{0}(\xi)+2\,i\,r(\xi)\,b'_{0}(\xi)+2\,\xi\,\big(\chi_{1}^{a}\big)'\big(-\psi'(\xi)\big)\,b_{1}(\xi)
\end{equation}
As a consequence,
\begin{equation*}
\mathcal{W}^{a}_{+}(\hat{u},\overline{\hat{u}})=
\int_{0}^{+\infty}d_{0}(\xi)\,\overline{b_{0}(\xi)}\,d\xi+h\,
\int_{0}^{+\infty}\big(d_{0}(\xi)\,\overline{b_{1}(\xi)}+
d_{1}(\xi)\,\overline{b_{0}(\xi)}\big)\,d\xi+\mathcal{O}(h^{2})
:=\mathcal{M}_{0}^{+}+h\,\mathcal{M}_{1}^{+}+\mathcal{O}(h^{2})
\end{equation*}
{\bf Computation of $\mathcal{M}_{0}^{+}$:} Using formula (\ref{E6}), we obtain
\begin{align*}
\mathcal{M}_{0}^{+}&=\int_{0}^{+\infty}2\,\xi\,\big(\chi_{1}^{a}\big)'\big(-\psi'(\xi)\big)\,|b_{0}(\xi)|^{2}\,d\xi\ \\
&=|C_{0}|^{2}\,\int_{0}^{+\infty}\frac{\alpha(\xi)}{|\alpha(\xi)|}\,
\psi''(\xi)\,\big(\chi_{1}^{a}\big)'\big(-\psi'(\xi)\big)\,d\xi=|C_{0}|^{2}\,\mathrm{sgn}\big(\alpha(0)\big)
\end{align*}
Here we used the fact that
\begin{equation*}
\lim_{\xi\rightarrow +\infty}\chi_{1}^{a}\big(-\psi'(\xi)\big)=0,\quad \chi_{1}^{a}\big(-\psi'(0)\big)=\chi_{1}^{a}(x_{E})=1
\end{equation*}
{\bf Computation of $\mathcal{M}_{1}^{+}$:} A straightforward computation shows that
\begin{equation}\label{dZEROunBARRe}
d_{0}(\xi)\,\overline{b_{1}(\xi)}=C_{0}\,\mathrm{sgn}(\alpha)\,\psi''\,\big(\chi_{1}^{a}\big)'\,
\big(\overline{C_{1}+D_{1}(\xi)}\big)
\end{equation}
\begin{equation}\label{dunZEROBBARRe}
d_{1}(\xi)\,\overline{b_{0}(\xi)}=
i\,|C_{0}|^{2}\,\frac{r'}{|\alpha|}-i\,r\,\frac{\alpha'}{\alpha}\,\frac{|C_{0}|^{2}}{|\alpha|}-i\,\chi''_{1}\,\frac{|C_{0}|^{2}}{|\alpha|}+
\overline{C_{0}}\,\mathrm{sgn}(\alpha)\,\psi''\,\big(\chi_{1}^{a}\big)'\,
\big(C_{1}+D_{1}(\xi)\big)
\end{equation}
which gives
\begin{equation}\label{dZEROPLUSBBARRe}
d_{0}(\xi)\,\overline{b_{1}(\xi)}+d_{1}(\xi)\,\overline{b_{0}(\xi)}=
2\,\text{Re}\Big[\overline{C_{0}}\,\mathrm{sgn}(\alpha)\,\psi''\,\big(\chi_{1}^{a}\big)'\,
\big(C_{1}+D_{1}(\xi)\big)\Big]+
i\,|C_{0}|^{2}\,\frac{r'}{|\alpha|}-i\,|C_{0}|^{2}\,\frac{r\alpha'}{\alpha|\alpha|}
\end{equation}
Hence,
\begin{equation}\label{Mathcalunplus}
\mathcal{M}_{1}^{+}=2\,\text{Re}(J_{1})+i\,|C_{0}|^{2}\,(J_{2}-J_{3})
\end{equation}
where we have set
\begin{equation}\label{Jindiceun}
J_{1}=\overline{C_{0}}\,\int_{0}^{+\infty}\text{sgn}(\alpha)\,\psi''\,\big(\chi_{1}^{a}\big)'\,
\big(C_{1}+D_{1}(\xi)\big)\,d\xi
\end{equation}
\begin{equation}\label{Jindicedeux}
J_{2}=\int_{0}^{+\infty}\frac{r'}{|\alpha|}\,d\xi,\quad J_{3}=\int_{0}^{+\infty}\frac{r\alpha'}{\alpha|\alpha|}\,d\xi
\end{equation}
We easily obtain
\begin{equation}\label{}
J_{1}=\mathrm{sgn}\big(\alpha(0)\big)\,\overline{C_{0}}\,C_{1}+\overline{C_{0}}\,\int_{0}^{+\infty}\mathrm{sgn}(\alpha)(\xi)\,\chi_{1}^{a}\,D'_{1}(\xi)\,d\xi
\end{equation}
Since $D'_{1}(\xi)$ is purely imaginary, we obtain
\begin{equation}\label{}
\text{Re}(J_{1})=\mathrm{sgn}\big(\alpha(0)\big)\,\text{Re}(\overline{C_{0}}\,C_{1})
\end{equation}
Moreover, by integrating by parts and using formula (\ref{FormuLERxi}), we get
\begin{equation}\label{}
\begin{aligned}
J_{2}-J_{3}&=\int_{0}^{+\infty}\frac{r'}{|\alpha|}\,d\xi-\int_{0}^{+\infty}\frac{r\alpha'}{\alpha|\alpha|}\,d\xi\ \\
           &=\Big[\frac{r(\xi)}{|\alpha(\xi)|}\Big]_{0}^{+\infty}=
\Big[\frac{\mathrm{sgn}(\alpha)\,\psi''(\xi)\,\big(\chi_{1}^{a}\big)''}{2}\Big]_{0}^{+\infty}=0
\end{aligned}
\end{equation}
Therefore,
\begin{equation}\label{Mathcalunplus}
\mathcal{M}_{1}^{+}=2\,\mathrm{sgn}\big(\alpha(0)\big)\,\text{Re}(\overline{C_{0}}\,C_{1}).
\end{equation}
It follows that
\begin{equation}
\mathcal{W}^{a}_{+}\big(\widehat{u},\overline{\widehat{u}}\big)=|C_{0}|^{2}\,\mathrm{sgn}\big(\alpha(0)\big)
+2\,h\,\mathrm{sgn}\big(\alpha(0)\big)\,\text{Re}(\overline{C_{0}}\,C_{1})+\mathcal{O}(h^{2})
\end{equation}
Similarly, one shows that
\begin{equation}
\mathcal{W}^{a}_{-}\big(\widehat{u},\overline{\widehat{u}}\big)=-|C_{0}|^{2}\,\mathrm{sgn}\big(\alpha(0)\big)
-2\,h\,\mathrm{sgn}\big(\alpha(0)\big)\,\text{Re}(\overline{C_{0}}\,C_{1})+\mathcal{O}(h^{2}),
\end{equation}
which allows us to conclude that, modulo $\mathcal{O}(h^{2})$,
\begin{equation}\label{}
\mathcal{W}^{a}\big(\widehat{u},\overline{\widehat{u}}\big):=
\mathcal{W}^{a}_{+}\big(\widehat{u},\overline{\widehat{u}}\big)-\mathcal{W}^{a}_{-}\big(\widehat{u},\overline{\widehat{u}}\big)=
2\,|C_{0}|^{2}\,\mathrm{sgn}\big(\alpha(0)\big)+4\,h\,\mathrm{sgn}\big(\alpha(0)\big)\,\text{Re}(\overline{C_{0}}\,C_{1})+
\mathcal{O}(h^{2})
\end{equation}
In the case where $\alpha(0)=V'(x_{E})>0$, with $C_{0}>0$ and $C_{1}\in \mathbb{R}$, we necessarily obtain
\begin{equation}\label{C00}
C_{0}=2^{-1/2},\quad C_{1}=0
\end{equation}
For these values of $C_{0}$ and $C_{1}$, one checks that
\begin{equation}\label{}
\mathcal{W}^{a}\big(\widehat{u},\overline{\widehat{u}}\big)=1+\mathcal{O}(h^{2})
\end{equation}
\begin{rem}
So far, we have obtained, still with $\alpha(0)>0$:
\begin{equation}\label{}
\begin{aligned}
\widehat{u}^{a}(\xi;h)&=e^{\frac{i}{h}\,\psi(\xi)}\,\Big(2^{-1/2}+h\,D_{1}(\xi)+\mathcal{O}(h^{2})\Big)\,|\alpha(\xi)|^{-\frac{1}{2}}\ \\
                      &=e^{\frac{i}{h}\,\psi(\xi)}\,2^{-1/2}\,\Big(1+i\,h\,\mathrm{Im}\big(D_{1}(\xi)\big)+\mathcal{O}(h^{2})\Big)\,|\alpha(\xi)|^{-\frac{1}{2}}
\end{aligned}
\end{equation}
where $\mathrm{Im}\big(D_{1}(\xi)\big)$ is given by (\ref{INTEGPARt0E0ZZ19}). Using the approximation $(h\rightarrow 0)$
\begin{equation}\label{}
1+i\,h\,\mathrm{Im}\big(D_{1}(\xi)\big)+\mathcal{O}(h^{2})=\exp\Big[i\,h\,\mathrm{Im}\big(D_{1}(\xi)\big)\Big]\,\big(1+\mathcal{O}(h^{2})\big),
\end{equation}
the normalized WKB solution near $a_{E}=(x_{E},0)$ now writes
\begin{equation}\label{UCHAPEAU}
\widehat{u}^{a}(\xi;h)=2^{-1/2}\,|\alpha(\xi)|^{-\frac{1}{2}}\,e^{i\widetilde{S}(\xi;h)/h}\,\big(1+\mathcal{O}(h^{2})\big)
\end{equation}
with the $h$-dependent phase function
\begin{equation}\label{}
\widetilde{S}(\xi;h)=\psi(\xi)+h^{2}\,\mathrm{Im}\big(D_{1}(\xi)\big)
\end{equation}
\end{rem}
\subsubsection{The homology class of the generalized action: Fourier representation}
Here we identify the various term in (\ref{UCHAPEAU}), which are responsible for the holonomy of $\widehat{u}^{a}$. First on $\gamma_{E}$ we have
$\psi(\xi)=\displaystyle\int-x\,d\xi+\mathrm{Const.}$ and $\varphi(x)=\displaystyle\int \xi\,dx+\mathrm{Const}$. By Hamilton equations
$$\dot{\xi}(t)=-\partial_{x} p_{0}\big(x(t),\xi(t)\big)=-V'\big(x(t)\big),\quad \dot{x}(t)=\partial_{\xi} p_{0}\big(x(t),\xi(t)\big)=2\,\xi(t)$$
where $t$ is the parametrization of $\gamma_{E}$ by the time evolution. Next we consider $D_{1}(\xi)$ as the integral over $\gamma_{E}$ of the 1-form $\Omega_{1}(\xi)$, defined near $\xi_{E}=0$ in Fourier representation by the expression (\ref{0E019}). Using WKB construction, $\Omega_{1}(\xi)$ can also be extended to the spatial representation. Since $\gamma_{E}$ is Lagrangian, $\Omega_{1}(\xi)$ is a closed form that we are going to compute modulo exact forms. Using integration by parts in (\ref{0E019}), together with the condition $D_{1}(0)=0$, we derive the following relations
\begin{equation}\label{}
\sqrt{2}\,\mathrm{Re}\big(D_{1}(\xi)\big)=0
\end{equation}
and
\begin{equation}\label{Racin2ImDind1}
\sqrt{2}\,\mathrm{Im}\big(D_{1}(\xi)\big)=\int_{0}^{\xi}T_{1}(\zeta)\,d\zeta+
\big[\frac{\psi''}{6\,\alpha}\,V'''+\frac{\alpha'}{4\,\alpha^{2}}\,V''\big]_{0}^{\xi}
\end{equation}
with
\begin{equation}\label{TINdice1}
T_{1}(\zeta)=\frac{1}{\alpha}\,
\Big[\frac{(\psi'')^{2}}{24}\,V^{(4)}+\frac{\alpha'\,\psi''}{6\,\alpha}\,
V'''+\frac{(\alpha')^{2}}{8\,\alpha^{2}}\,V''\Big]_{x=-\psi'(\zeta)}
\end{equation}
There follows:
\begin{lem}\label{LemmMa10101}
Modulo the integral of an exact form in $\mathcal{A}$, with $T_{1}$ as in (\ref{TINdice1}) we have:
\begin{equation}\label{}
\mathrm{Re}\big(D_{1}(\xi)\big)\equiv0,\quad \sqrt{2}\,\mathrm{Im}\big(D_{1}(\xi)\big)\equiv\int_{0}^{\xi}T_{1}(\zeta)\,d\zeta
\end{equation}
\end{lem}
If $u(x,\xi)$, $v(x,\xi)$ are any smooth functions on $\mathcal{A}$ we set $\Omega(x,\xi)=u(x,\xi)\,dx+v(x,\xi)\,d\xi$. By Stokes' formula
\begin{equation*}
\int_{\gamma_{E}}\Omega(x,\xi)=\int\int_{\{\xi^{2}+V(x)\leq E\}}\big(\partial_{x} v-\partial_{\xi} u\big)\,dx \wedge{d\xi}
\end{equation*}
Making the symplectic change of coordinates $(x,\xi)\mapsto (t,E)$ in $T^{*}\mathbb{R}$
\begin{equation}\label{St2}
\int\int_{\{p_{0}\leq E\}}\big(\partial_{x} v-\partial_{\xi} u\big)\,dx \wedge{d\xi}=\int_{0}^{E}\int_{0}^{T(E')}\big(\partial_{x} v-\partial_{\xi} u\big)\,dt \wedge{dE'}
\end{equation}
where $T(E')$ is the period of the flow of Hamilton vector field $H_{p_{0}}$ at energy $E'$ \big($T(E')$ being a constant near (0,0)\big). Taking derivative with respect to $E$, we find:
\begin{equation}\label{FORMULASTOKESS}
\frac{d}{dE}\,\int_{\gamma_{E}}\Omega(x,\xi)=\int_{0}^{T(E)}\big(\partial_{x} v-\partial_{\xi} u\big)\,dt
\end{equation}
We want to simplify the expression (\ref{INTEGPARt0E0ZZ19}). For this purpose, we set
\begin{equation}\label{}
g_{0}(x,\xi)=\frac{V''(x)}{V'(x)}\,\big(\psi'''(\xi)\,\alpha(\xi)+\psi''(\xi)\,\alpha'(\xi)\big)=2\,\frac{V''(x)}{V'(x)}
\end{equation}
Taking second derivative of eikonal equation (\ref{Jacob}), we get
\begin{equation*}
\frac{(\partial_{x} g_{0})\big(-\psi'(\zeta),\zeta\big)}{\alpha(\zeta)}
=\frac{\psi'''\,V'''}{\alpha}+\frac{\psi''\,\alpha'\,V'''}{\alpha^{2}}-\frac{\psi'''\,(V'')^{2}}{\alpha^{2}}+
\frac{(\alpha')^{2}\,V''}{\alpha^{3}}
\end{equation*}
By integration by parts, we obtain
\begin{equation*}
\int_{0}^{\xi}\frac{(\partial_{x} g_{0})\big(-\psi'(\zeta),\zeta\big)}{\alpha(\zeta)}\,d\zeta=
\int_{0}^{\xi}\frac{(\psi'')^{2}\,V^{(4)}}{\alpha}\,d\zeta+
3\,\int_{0}^{\xi}\frac{(\alpha')^{2}\,V''}{\alpha^{3}}\,d\zeta+4\,
\int_{0}^{\xi}\frac{\psi''\,\alpha'\,V'''}{\alpha^{2}}\,d\zeta+
\big[\frac{\psi''\,V'''}{\alpha}+
\frac{\alpha'\,V''}{\alpha^{2}}\big]_{0}^{\xi}
\end{equation*}
so
\begin{equation}\label{}
\frac{24}{C_{0}}\,\text{Im}\big(D_{1}(\xi)\big)=\int_{0}^{\xi}\frac{(\partial_{x} g_{0})\big(-\psi'(\zeta),\zeta\big)}{\alpha(\zeta)}\,d\zeta
+\big[\frac{3\,\psi''\,V'''}{\alpha}+
\frac{5\,\alpha'\,V''}{\alpha^{2}}\big]_{0}^{\xi}
\end{equation}
Hence, modulo exact forms, we have
\begin{align*}
\frac{1}{C_{0}}\,\text{Im}\big(D_{1}(\xi)\big)&\equiv\frac{1}{24}\,\int_{0}^{\xi}\frac{(\partial_{x}g_{0})\big(-\psi'(\zeta),\zeta\big)}{\alpha(\zeta)}\,d\zeta
=-\frac{1}{24}\,\int_{0}^{T(E)}\partial_{x} g_{0}\big(x(t),\xi(t)\big)\,dt\ \\
     &=-\frac{1}{24}\,\frac{d}{dE}\,\int_{\gamma_{E}}g_{0}(x,\xi)\,d\xi
     =-\frac{1}{12}\,\frac{d}{dE}\,\int_{\gamma_{E}}\frac{V''(x)}{\partial_{x} p_{0}(x,\xi)}\,\,d\xi\ \\
     &=\frac{1}{12}\,\frac{d}{dE}\,\int_{0}^{T(E)}V''\big(x(t)\big)\,dt
\end{align*}
Writing
$$\sqrt{2}\,D_{1}(\xi)=\int_{0}^{\xi}\Omega_{1}(\zeta)$$
we find that
\begin{equation}
\mathrm{Im}\oint_{\gamma_{E}}\Omega_{1}=
\frac{1}{12}\,\frac{d}{dE}\,\int_{0}^{T(E)}V''\big(x(t)\big)\,dt
\end{equation}
Since
\begin{equation}
\text{Re}\big(D_{1}(\xi)\big)=0
\end{equation}
we deduce that
\begin{equation}\label{}
\mathrm{Re}\oint_{\gamma_{E}}\Omega_{1}=0
\end{equation}
\subsubsection{Asymptotic solutions in the spatial representation}
In order to construct the microlocal solution of the stationary Schr\"{o}dinger equation
\begin{equation}\label{bggHkw1}
\big((hD_{x})^{2}+V(x)\big)u(x)=E\,u(x)
\end{equation}
in the spatial representation up to second order in $h$, it suffices to apply the inverse semi-classical Fourier transform to the solution $\widehat{u}^{a}(\xi;h)$:
\begin{equation}\label{FourieRInveRseSemi}
u^{a}(x;h)=(2\,\pi\,h)^{-\frac{1}{2}}\,\int e^{\frac{i}{h}\,x\,\xi}\,\widehat{u}^{a}(\xi;h)\,d\xi=(2\,\pi\,h)^{-\frac{1}{2}}\,\int e^{\frac{i}{h}\,(x\,\xi+\psi(\xi))}\,b(\xi;h)\,d\xi
\end{equation}
In general, the integral in (\ref{FourieRInveRseSemi}) cannot be computed explicitly. However, since the WKB wave function is itself an approximation, one may obtain an asymptotic evaluation. This is achieved using the method of stationary phase.

For fixed $x$, the phase function associated with the oscillatory integral (\ref{FourieRInveRseSemi}) defining $u^{a}(x;h)$ is
\begin{equation}\label{FoNcTiOnPhAsE}
g_{x}(\xi)=x\,\xi+\psi(\xi).
\end{equation}
It admits two critical points, $\xi_{+}(x)>0$ and $\xi_{-}(x)<0$, which are the roots of
\begin{equation}\label{EqUatiOPOICRItiQQ}
x+\psi'\big(\xi_{\pm}(x)\big)=0.
\end{equation}
The corresponding critical values of $g_{x}$ are
\begin{equation}\label{VALEURtiOPOICRItiQQ}
\varphi_{\pm}(x):=g_{x}\big(\xi_{\pm}(x)\big)=x\,\xi_{\pm}(x)+\psi\big(\xi_{\pm}(x)\big).
\end{equation}
Differentiating (\ref{VALEURtiOPOICRItiQQ}) with respect to $x$ yields
\begin{equation}\label{}
\varphi'_{\pm}(x)=\xi_{\pm}(x)+\Big[x+\psi'\big(\xi_{\pm}(x)\big)\Big]\,\xi_{\pm}(x)=\xi_{\pm}(x)
\end{equation}
Since $\xi_{\pm}(x_{E})=0$ and $\psi(0)=0$, it follows that $\varphi_{\pm}(x_{E})=0$, and therefore
\begin{equation}\label{}
\varphi_{\pm}(x)=\int_{x_{E}}^{x}\xi_{\pm}(y)\,dy:=\varphi_{\pm}(x_{E}, x).
\end{equation}
As the phase function has two critical points contained in the support of $|\alpha(\xi)|^{-1/2}$, one must sum the contributions from both. By the stationary phase theorem (\ref{Phasestat0}), we obtain
\begin{equation}\label{}
\int e^{\frac{i}{h}\big(x\,\xi+\psi(\xi)\big)}\,b(\xi;h)\,d\xi\sim \sum_{\pm} e^{\frac{i}{h}\,\varphi_{\pm}(x_{E}, x)}\,
\bigg(\frac{\Big|\psi''\big(\xi_{\pm}(x)\big)\Big|}{2\,\pi\,h}\bigg)^{-1/2}\,e^{\pm i\,\frac{\pi}{4}}\,\sum_{j}h^{j}L_{j}b\big(\xi_{\pm}(x);h\big)
\end{equation}
Using (\ref{E6}), this becomes
\begin{equation}\label{}
u^{a}(x;h)\sim 2^{-1/2}\,\sum_{\pm} e^{\frac{i}{h}\,\varphi_{\pm}(x_{E}, x)}\,
\big|\xi_{\pm}(x)\big|^{-1/2}\,\Big|\alpha\big(\xi_{\pm}(x)\big)\Big|^{1/2}\,e^{\pm i\,\frac{\pi}{4}}\,\sum_{j}h^{j}L_{j}b\big(\xi_{\pm}(x);h\big)
\end{equation}
A direct computation using (\ref{TTL551A06}) shows that
$$L_{1}b\big(\xi_{\pm}(x);h\big)=\sum_{n=0}^{2}\frac{2^{-(n+1)}}{i\,n!\,(n+1)!}\,
\Big[-\big(\psi''(\xi_{\pm}(x))\big)^{-1}\,\frac{\partial^{2}}{\partial \xi^{2}}\Big]^{n+1}(\Phi^{n}_{\pm, x}\,b)\big(\xi_{\pm}(x);h\big)$$
with
\begin{equation}\label{}
\begin{aligned}
\Phi_{\pm, x}(\xi)&=g_{x}(\xi)-g_{x}\big(\xi_{\pm}(x)\big)
-\frac{1}{2}\,\big(\xi-\xi_{\pm}(x)\big)^{2}\,g''_{x}\big(\xi_{\pm}(x)\big)\ \\
&=x\,\big(\xi-\xi_{\pm}(x)\big)+\psi(\xi)-\psi\big(\xi_{\pm}(x)\big)
-\frac{1}{2}\,\big(\xi-\xi_{\pm}(x)\big)^{2}\,\psi''\big(\xi_{\pm}(x)\big)=\mathcal{O}\Big(\big(\xi-\xi_{\pm}(x)\big)^{3}\Big)
\end{aligned}
\end{equation}
After further computation, one finds
$$L_{1}b=\frac{i}{2}\,(\psi'')^{-1}\,\frac{\partial^{2}b}{\partial \xi^{2}}-\frac{i}{8}\,
(\psi'')^{-2}\,\big(4\,\psi'''\,\partial_{\xi} b+\psi^{(4)}\,b\big)+\frac{5\,i}{24}\,(\psi'')^{-3}\,(\psi''')^{2}\,b$$
and therefore
\begin{equation}\label{}
\begin{aligned}
b+h\,L_{1}b&=b_{0}\,
\big[1+h\,\frac{b_{1}}{b_{0}}+
\frac{i\,h}{2}\,(\psi'')^{-1}\,\frac{b''_{0}}{b_{0}}-\frac{i\,h}{8}\,
(\psi'')^{-2}\,\big(4\,\psi'''\,\frac{b'_{0}}{b_{0}}+\psi^{(4)}\big)+\frac{5\,i\,h}{24}\,(\psi'')^{-3}\,(\psi''')^{2}\big]+\mathcal{O}(h^{2})\ \\
&=b_{0}\,\big[1+i\,h\,\tilde{D}_{1}+\frac{i\,h}{2}\,(\psi'')^{-1}\,(\theta_{0}^{2}-\theta'_{0})-
\frac{i\,h}{8}\,
(\psi'')^{-2}\,\big(\psi^{(4)}-4\,\theta_{0}\,\psi'''\big)+\frac{5\,i\,h}{24}\,(\psi'')^{-3}\,(\psi''')^{2}\big]+\mathcal{O}(h^{2})
\end{aligned}
\end{equation}
In a neighborhood of the focal point $a=(x_{E},0)$ and for $x<x_{E}$, the microlocal solution of the stationary Schr\"{o}dinger equation (\ref{bggHkw1}) is given by
\begin{equation}\label{uaplusplusOINS}
u^{a}(x;h)=\sum_{\pm}u^{a}_{\pm}(x;h)
=\frac{1}{2}\,\sum_{\pm}\,e^{\frac{i}{h}\,\varphi_{\pm}(a, x)}\,\big|\xi_{\pm}(x)\big|^{-\frac{1}{2}}\,e^{\pm i\,\frac{\pi}{4}}\,
\Big[1+i\,h\,Z\big(\xi_{\pm}(x)\big)+\mathcal{O}(h^{2})\Big]
\end{equation}
where
\begin{equation}\label{ZPuaplusplusOINS}
Z\big(\xi_{\pm}(x)\big)=\mathrm{Im}\Big(D_{1}\big(\xi_{\pm}(x)\big)\Big)+R\big(\xi_{\pm}(x)\big),
\end{equation}
with
\begin{equation}\label{RPuaplusplusOINS}
R(\xi)=\frac{1}{2}\,\big(\psi''(\xi)\big)^{-1}\,\Big[\theta_{0}^{2}(\xi)-\theta'_{0}(\xi)\Big]-\frac{1}{8}\,
\big(\psi''(\xi)\big)^{-2}\,\big[\psi^{(4)}(\xi)-4\,\theta_{0}(\xi)\,\psi'''(\xi)\big]+\frac{5}{24}\,\big(\psi''(\xi)\big)^{-3}\,\big(\psi'''(\xi)\big)^{2},
\end{equation}
The phase function is
\begin{equation}\label{RPuaplhusplusOINS}
\varphi_{\pm}(a,x)=\int_{a}^{x}\xi_{\pm}(y)\,dy
\end{equation}
First, we observe that $R\big(\xi_{\pm}(x)\big)$ does not contribute to the homology class of the semi-classical forms defining the action, as it contains no integral term. Thus, the phase in (\ref{uaplusplusOINS}) can be replaced, modulo $\mathcal{O}(h^{3})$ by
\begin{equation}\label{Phasesph5atiale}
S_{\pm}(a,x;h)=\varphi_{\pm}(a,x)+h^{2}\,\mathrm{Im}\Big(D_{1}\big(\xi_{\pm}(x)\big)\Big)
=\int_{a}^{x}\xi_{\pm}(y)\,dy+h^{2}\,\int_{a}^{x}T_{1}\big(\xi_{\pm}(y)\big)\,\xi'_{\pm}(y)\,dy
\end{equation}
where $T_{1}$ is given by (\ref{TINdice1}).
\paragraph{}
Away from $x_{E}$, we use standard WKB theory extending (\ref{uaplusplusOINS}), with the Ansatz
\begin{equation}\label{ANSATZ}
u^{a}_{\rho}(x;h)=\widetilde{a}_{\rho}(x;h)\,e^{\frac{i}{h}\,\varphi_{\rho}(x)}
\end{equation}
where $\widetilde{a}_{\rho}(x;h)$ is a formal series in $h$. For instance, the most suitable choice here is
\begin{equation*}
\widetilde{a}_{\rho}(x;h)=\widetilde{a}_{\rho,0}(x)+h\,\widetilde{a}_{\rho,1}(x)+h^{2}\,\widetilde{a}_{\rho,2}(x)+\cdots
\end{equation*}
The phase $\varphi_{\rho}(x)$ is real, smooth, and must satisfy the eikonal equation
\begin{equation}\label{eico0}
p_{0}\big(x,\varphi'_{\rho}(x)\big)=E\Leftrightarrow \big(\varphi'_{\rho}(x)\big)^{2}+V(x)=E
\end{equation}
Applying the operator $P(x,hD_{x})-E$ to the Ansatz (\ref{ANSATZ}), one finds
\begin{equation}\label{Haper3}
e^{-\frac{i}{h}\,\varphi_{\rho}}\,\big(P(x,hD_{x})-E\big)\,\big(\widetilde{a}_{\rho}\,e^{\frac{i}{h}\,\varphi_{\rho}}\big)=(p_{0}-E)\,\widetilde{a}_{\rho}
-i\,h\,\big(\varphi''_{\rho}\,\widetilde{a}_{\rho}+2\,\varphi'_{\rho}\,\partial_{x}\widetilde{a}_{\rho}\big)-h^{2}\,\frac{\partial^{2} \widetilde{a}_{\rho}}{\partial x^{2}}
\end{equation}
Omitting indices $\rho=\pm$ and $a$, we construct formal solutions (i.e. in the sense of classical formal symbols) of
$$\big(P(x,hD_{x})-E\big)\big(\widetilde{a}(x;h)\,e^{\frac{i}{h}\,\varphi(x)}\big)=0$$
and provided that the eikonal equation (\ref{eico0}) holds, one obtains the first transport equation
\begin{equation}\label{A0zero}
2\,\varphi'(x)\,\widetilde{a}'_{0}(x)+\varphi''(x)\,\widetilde{a}_{0}(x)=0
\end{equation}
whose solutions are given by
\begin{equation*}
\widetilde{a}_{0}(x)=\widetilde{C}_{0}\,|\varphi'(x)|^{-\frac{1}{2}}
\end{equation*}
By annihilating the $h^{2}$-term in (\ref{Haper3}), one shows that the subprincipal term $\widetilde{a}_{1}(x)$ satisfies the differential equation
\begin{equation}\label{A1un}
2\,\varphi'(x)\,\widetilde{a}'_{1}(x)+\varphi''(x)\,\widetilde{a}_{1}(x)=i\,\widetilde{a}''_{0}(x)
\end{equation}
The homogeneous equation associated with (\ref{A1un}) is the same as that for $\widetilde{a}_{0}$ (\ref{A0zero}); hence we look for a particular solution of (\ref{A1un}) of the form
$$\widetilde{D}_{1}(x)\,|\varphi'(x)|^{-1/2}$$
By the method of variation of constants, one obtains
\begin{equation}\label{DTILDEone}
\widetilde{D}_{1}(x)=\frac{i}{2}\,\int^{x}\mathrm{sgn}\big(\varphi'(y)\big)\,|\varphi'(y)|^{-\frac{1}{2}}\,\widetilde{a}''_{0}(y)\,dy
=-i\,\frac{\widetilde{C}_{0}}{4}\,\int^{x}\frac{\{\varphi,y\}}{\varphi'(y)}\,dy
\end{equation}
where $\{\varphi,y\}$ denotes the Schwarzian derivative of the phase function $\varphi$ with respect to the real variable $y$, defined as
\begin{equation}\label{DriveSchwarZienne}
\{\varphi,y\}=\big(\frac{\varphi''}{\varphi'}\big)^{'}(y)-\frac{1}{2}\,\big(\frac{\varphi''(y)}{\varphi'(y)}\big)^{2}
\end{equation}
The function $\widetilde{D}_{1}$ may be normalized by $\widetilde{D}_{1}(x_{E})=0$, thus, the general solution of (\ref{A1un}) is
\begin{equation}\label{bkw10}
\widetilde{a}_{1}(x)=\Big(\widetilde{C}_{1}-i\,\frac{\widetilde{C}_{0}}{4}\,
\int^{x}_{x_{E}}\frac{\{\varphi,y\}}{\varphi'(y)}\,dy\Big)\,\big|\varphi'(x)\big|^{-1/2}
\end{equation}
The function $\widetilde{a}_{1}$ is determined by its value at $x_{E}$. Consequently,
$$a(x;h)=\Big(\widetilde{C}_{0}+h\,\widetilde{C}_{1}+h\,\widetilde{D}_{1}(x)+\mathcal{O}(h^{2})\Big)\,|\varphi'(x)|^{-\frac{1}{2}}.$$
Repeating this construction for the other branch $(\rho=-1)$ yields the two branches of WKB solutions:
\begin{equation}\label{2.10}
u^{a}_{\pm}(x;h)=\Big(\widetilde{C}_{0}+h\,\widetilde{C}_{1}+h\,\widetilde{D}_{1}^{\pm}(x)+\mathcal{O}(h^{2})\Big)\,
|\varphi'_{\pm}(x)|^{-\frac{1}{2}}\,e^{\frac{i}{h}\varphi_{\pm}(x)},
\end{equation}
where
\begin{equation}\label{2.11}
\widetilde{D}_{1}^{\pm}(x)=-i\,\frac{\widetilde{C}_{0}}{4}\,\int^{x}_{x_{E}}\frac{\{\varphi_{\pm},y\}}{\varphi'_{\pm}(y)}\,dy
\end{equation}
Here we have used that $\varphi_{\pm}(x)=\displaystyle\int_{x_{E}}^{x}\xi_{\pm}(y)\,dy$, with $p_0\big(x,\xi_{\pm}(x)\big)=E$.

Normalization with respect to the "flux norm" consists as above in computing $F^{a}_{\pm}=\displaystyle\frac{i}{h}\,[P,\chi^{a}]_{\pm}u^{a}$ by stationary phase theorem (\ref{Phasestat0}) modulo $\mathcal{O}(h^{2})$. Assuming already $\widetilde{C}_{0}$, $\widetilde{C}_{1}$ to be real, a simple calculation using integration by parts yields $\widetilde{C}_{0}=C_{0}=2^{-1/2}$, and $\widetilde{C}_{1}=C_{1}=0$. As a result, outside any neighborhood of $x_{E}$, we have
\begin{equation}\label{2.12}
u_{\pm}(x;h)=
2^{-1/2}\,\big(1+h\,\sqrt{2}\,\widetilde{D}^{\pm}_{1}(x)+\mathcal{O}(h^{2})\big)\,|\varphi'_{\pm}(x)|^{-\frac{1}{2}}\,e^{\frac{i}{h}\varphi_{\pm}(x)}.
\end{equation}
From (\ref{2.12}) we can recover the homology class of generalized action, considering the superposition $u(x;h)=e^{i\pi/4}\,u_{+}(x;h)+e^{-i\pi/4}\,u_{-}(x;h)$ near $a$. The argument is then similar to that of \cite{AIfaRrou}, formula (1).
\subsection{Quantization Rules and Fredholm Property}
Near the focal point $a=(x_{E},0)$, the microlocal solution of $(P-E)u=0$ takes the form
\begin{equation}\label{QRFP1}
u^{a}(x;h)=\sum_{\pm}u^{a}_{\pm}(x;h)=\frac{1}{2}\,\sum_{\pm}\,e^{\pm i\,\frac{\pi}{4}}\,e^{\frac{i}{h}\,S_{\pm}(a,x;h)}\,\big|\xi_{\pm}(x)\big|^{-\frac{1}{2}}\,
\big(1+\mathcal{O}(h^{2})\big)
\end{equation}
where $S_{\pm}(a,x;h)$ is defined in (\ref{Phasesph5atiale}). Similarly, near $a'=(x'_{E},0)$ one has
\begin{equation}\label{QRFP2}
u^{a'}(x;h)=\sum_{\pm}u^{a'}_{\pm}(x;h)=\frac{1}{2}\,\sum_{\pm}\,e^{\mp i\,\frac{\pi}{4}}\,e^{\frac{i}{h}\,S_{\pm}(a',x;h)}\,\big|\xi_{\pm}(x)\big|^{-\frac{1}{2}}\,
\big(1+\mathcal{O}(h^{2})\big)
\end{equation}
The change of phase factor $e^{\pm i\pi/4}$ accounts for \emph{Maslov index}. For the sake of simplicity, we omit henceforth $\mathcal{O}(h^{2})$ terms, but the computations above extend to all order in $h$, thus giving the asymptotics of BS.

The semiclassical distributions $u^{a}, u^{a'}$ span the microlocal kernel $\mathcal{K}_{h}(E)$ of $P-E$ in $(x,\xi)\in ]a',a[\times\mathbb{R}$, and are normalized using microlocal Wronskians as follows.

Let $\chi^a\in C_{0}^{\infty}(\mathbb{R}^{2})$ be a cut-off function as defined in Subsection \ref{nOrMaLiSaTiOnn}. We have
\begin{equation}\label{QRFP3}
\frac{i}{h}\,[P,\chi^a]=-i\,h\,\chi_{2}(\xi)\,\big(\chi_{1}^{a}\big)''(x)+2\,\chi_{2}(\xi)\,\big(\chi_{1}^{a}\big)'(x)\,hD_{x}
\end{equation}
Using (\ref{QRFP1})-(\ref{QRFP2}), modulo $\mathcal{O}(h^{2})$ we obtain
\begin{equation*}
\frac{i}{h}[P,\chi^{a}]u^{a}(x;h)=\sum_{\pm}\pm \big(\chi_{1}^{a}\big)'(x)\,|\xi_{\pm}(x)|^{1/2}\,e^{\frac{i}{h}\,S_{\pm}(a,x;h)}\,e^{\pm i\,\frac{\pi}{4}}-\frac{i\,h}{2}\,\sum_{\pm}
\frac{d}{dx}\Big[\big(\chi_{1}^{a}\big)'(x)\,|\xi_{\pm}(x)|^{-1}\Big]\,|\xi_{\pm}(x)|^{1/2}\,e^{\frac{i}{h}\,S_{\pm}(a,x;h)}\,e^{\pm i\,\frac{\pi}{4}}
\end{equation*}
\begin{equation*}
\frac{i}{h}[P,\chi^{a'}]u^{a'}(x;h)=\sum_{\pm}\pm \big(\chi_{1}^{a'}\big)'(x)\,|\xi_{\pm}(x)|^{1/2}\,e^{\frac{i}{h}\,S_{\pm}(a',x;h)}\,e^{\mp i\,\frac{\pi}{4}}-\frac{i\,h}{2}\,\sum_{\pm}
\frac{d}{dx}\Big[\big(\chi_{1}^{a'}\big)'(x)\,|\xi_{\pm}(x)|^{-1}\Big]\,|\xi_{\pm}(x)|^{1/2}\,e^{\frac{i}{h}\,S_{\pm}(a',x;h)}\,e^{\mp i\,\frac{\pi}{4}}
\end{equation*}
Let
\begin{equation}\label{QRFP4}
F^{a}_{\pm}(x;h):=\frac{i}{h}[P,\chi^{a}]_{\pm}u^{a}(x;h)
\end{equation}
Then, mod $\mathcal{O}(h^{2})$
\begin{equation*}
F^{a}_{\pm}(x;h)=\pm \big(\chi_{1}^{a}\big)'(x)\,|\xi_{\pm}(x)|^{1/2}\,e^{\frac{i}{h}\,S_{\pm}(a,x;h)}\,e^{\pm i\,\frac{\pi}{4}}-\frac{i\,h}{2}\,
\frac{d}{dx}\Big[\big(\chi_{1}^{a}\big)'(x)\,|\xi_{\pm}(x)|^{-1}\Big]\,|\xi_{\pm}(x)|^{1/2}\,e^{\frac{i}{h}\,S_{\pm}(a,x;h)}\,e^{\pm i\,\frac{\pi}{4}}
\end{equation*}
Hence
\begin{align*}
\big(F^{a}_{+}-F^{a}_{-}|u^{a}\big)&=\int_{-\infty}^{a}\big(\chi_{1}^{a}\big)'(x)\,dx+\frac{ih}{4}\,\int_{-\infty}^{a}
\frac{d}{dx}\Big[\big(\chi_{1}^{a}\big)'(x)\,|\xi_{-}(x)|^{-1}\Big]\,dx-\frac{ih}{4}\,\int_{-\infty}^{a}
\frac{d}{dx}\Big[\big(\chi_{1}^{a}\big)'(x)\,|\xi_{+}(x)|^{-1}\Big]\,dx=\int_{-\infty}^{a}\big(\chi_{1}^{a}\big)'(x)\,dx\ \\
&=\chi_{1}^{a}(a)-\lim_{x\rightarrow -\infty}\chi_{1}^{a}(x)=1
\end{align*}
Here we have used the fact that
\begin{align*}
\int_{-\infty}^{a}\frac{d}{dx}\Big[\big(\chi_{1}^{a}\big)'(x)\,|\xi_{\pm}(x)|^{-1}\Big]\,dx&=
\Big[\frac{\big(\chi_{1}^{a}\big)'(x)}{|\xi_{\pm}(x)|}\Big]_{-\infty}^{a}=\lim_{x\rightarrow a^{-}}\frac{\big(\chi_{1}^{a}\big)'(x)}{|\xi_{\pm}(x)|}-\lim_{x\rightarrow -\infty}\frac{\big(\chi_{1}^{a}\big)'(x)}{|\xi_{\pm}(x)|}\ \\
&=\lim_{x\rightarrow a^{-}}\frac{-2\,\big(\chi_{1}^{a}\big)''(x)\,\big(E-V(x)\big)^{1/2}}{V'(x)}=0
\end{align*}
Note that the mixed terms such as $(F^{a}_{\pm}|u^{a}_{\mp})$ are $\mathcal{O}(h^\infty)$ because the phase is non-stationary, and hence the microlocal solution $u^{a}$ is normalized modulo $\mathcal{O}(h^{2})$. In the same way, with
\begin{equation}\label{QRFP5}
\begin{aligned}
F^{a'}_{\pm}(x;h)&:=\frac{i}{h}[P,\chi^{a'}]_{\pm}\,u^{a'}(x,h)\ \\
                 &=\pm \big(\chi_{1}^{a'}\big)'(x)\,|\xi_{\pm}(x)|^{1/2}\,e^{\frac{i}{h}\,S_{\pm}(a',x;h)}\,e^{\mp i\,\frac{\pi}{4}}-\frac{i\,h}{2}\,
\frac{d}{dx}\Big[\big(\chi_{1}^{a'}\big)'(x)\,|\xi_{\pm}(x)|^{-1}\Big]\,|\xi_{\pm}(x)|^{1/2}\,e^{\frac{i}{h}\,S_{\pm}(a',x;h)}\,e^{\mp i\,\frac{\pi}{4}}
\end{aligned}
\end{equation}
the functions $F^{a}_+-F^{a}_-$ and $F^{a'}_+-F^{a'}_-$ span the microlocal co-kernel $\mathcal{K}^{*}_{h}(E)$ of $P-E$ in $(x,\xi)\in ]a',a[\times \mathbb{R}$. Microlocal solutions $u^{a}$ and $u^{a'}$ extend as smooth solutions on the whole interval $]a', a[$; we denote them by $u_{1}$ and $u_{2}$. Since there are no other focal points between $a$ and $a'$, they are expressed by the same formulae (which makes the analysis particularly simple) and satisfy:
$$\big(F^{a}_+-F^{a}_-|u_{1}\big)=1,\quad \big(F^{a'}_+-F^{a'}_-|u_{2}\big)=-1$$
We then compute \big(modulo $\mathcal{O}(h^{2})$\big)
\begin{equation}\label{QRFP6}
\big(F^{a'}_{+}-F^{a'}_{-}|u_{1}\big)=\frac{i}{2}\,\Big(e^{-\frac{i}{h}\,A_{+}(a,a';h)}-e^{-\frac{i}{h}\,A_{-}(a,a';h)}\Big)
\end{equation}
\big(taking again into account that the mixed terms are $\mathcal{O}(h^{\infty})$\big). Similarly
\begin{equation}\label{QRFP7}
\big(F^{a}_{+}-F^{a}_{-}|u_{2})=\frac{i}{2}\,\Big(e^{\frac{i}{h}\,A_{+}(a,a';h)}-e^{\frac{i}{h}\,A_{-}(a,a';h)}\Big)
\end{equation}
where we have set
\begin{equation}\label{QRFP8}
A_{\pm}(a,a';h):=S_{\pm}(a,x;h)-S_{\pm}(a',x;h)=\int_{a}^{a'}\xi_{\pm}(x)\,dx+
h^{2}\,\int_{a}^{a'}T_{1}\big(\xi_{\pm}(x)\big)\,\xi'_{\pm}(x)\,dx
\end{equation}
So we form Gram matrix
\begin{equation}\label{QRFP9}
G^{(a,a')}(E):=\left(
       \begin{array}{cc}
         \big(F^{a}_{+}-F^{a}_{-}|u_{1}\big)   &  \big(F^{a}_{+}-F^{a}_{-}|u_{2}\big) \\
         \big(F^{a'}_{+}-F^{a'}_{-}|u_{1}\big) & \big(F^{a'}_{+}-F^{a'}_{-}|u_{2}\big)\\
       \end{array}
     \right)
\end{equation}
whose determinant
$$D(E;h):=-\cos^{2}\Big(\frac{A_{-}(a,a';h)-A_{+}(a,a';h)}{2h}\Big)$$
vanishes precisely at the eigenvalues of $P$ in $I$, which yields, modulo $\mathcal{O}(h^{3})$,
\begin{equation}\label{QRFP10}
\int_{a'}^{a}\big(\xi_{+}(x)-\xi_{-}(x)\big)\,dx-\pi\,h
+h^{2}\,\int_{a'}^{a}\Big(T_{1}\big(\xi_{+}(x)\big)\,\xi'_{+}(x)-
T_{1}\big(\xi_{-}(x)\big)\,\xi'_{-}(x)\Big)\,dx=2n\pi h,\quad n\in\mathbb{Z}
\end{equation}
We next simplify (\ref{QRFP10}). First, we have
$$\int_{a'}^{a}\big(\xi_{+}(x)-\xi_{-}(x)\big)\,dx=\oint_{\gamma_{E}} \big(E-V(x)\big)^{1/2}\,dx$$
Moreover,
\begin{align*}
\int_{a'}^{a}\Big(T_{1}\big(\xi_{+}(x)\big)\,\xi'_{+}(x)-
T_{1}\big(\xi_{-}(x)\big)\,\xi'_{-}(x)\Big)\,dx&=\oint_{\gamma_{E}}T_{1}\big(\xi(x)\big)\,\xi'(x)\,dx\ \\
&=\mathrm{Im}\oint_{\gamma_{E}}\Omega_{_{1}}\big(\xi(x)\big)\ \\
&=\frac{1}{12}\,\frac{d}{dE}\int_{0}^{T(E)}V''\big(x(t)\big)\,dt
\end{align*}
Hence, the Bohr–Sommerfeld quantization rule up to order $h^{2}$ takes the form
\begin{equation}\label{QRFP11}
\oint_{\gamma_{E}}\big(E-V(x)\big)^{1/2}\,dx-\pi\,h
+\frac{h^{2}}{12}\,\frac{d}{dE}\int_{\gamma_{E}}V''\big(x(t)\big)\,dt
=2\,\pi\,n\,h+\mathcal{O}(h^{3})
\end{equation}
\section{Bohr-Sommerfeld and Action-Angle Variables}\label{ACTIO0NANGLE}
In this section we present a simple approach based on the Birkhoff normal form and the monodromy operator (see \cite{HR84}). Assume that the principal symbol $p_{0}$ of $P$ has a periodic orbit $\gamma_{0}\subset \mathcal{A}$ with period $2\pi$ at a non-critical energy $E=E_{0}$. It is known from \cite{Argyres} that there exists a smooth canonical transformation
$$(t,\tau)\mapsto \kappa (t,\tau)=(x,\xi),\quad t\in [0,2 \pi],$$
defined in a neighborhood of $\gamma_{0}$, and a smooth function $\tau\mapsto f_{0}(\tau)$, with $f_{0}(0)=0$ and $f'_{0}(0)=1$, such that
\begin{equation}\label{FNBK0}
p_{0}\circ \kappa(t,\tau)=f_{0}(\tau).
\end{equation}
This transformation is given by the generating function
$$S(\tau,x)=\int_{x_{E}}^{x}\xi(y),\quad \xi(x)=\partial_{x}S(\tau,x),\quad \varphi=\partial_{\tau}S(\tau,x),$$
with
$$p_{0}\big(x,\partial_{x}S(\tau,x)\big)=f_{0}(\tau).$$
The energy $E$ and the momentum $\tau$ are related by the transformation $E=f_{0}(\tau)$, with $f'_{0}(E_{0})=1$. The canonical map $\kappa$ can thus be quantized so as to reduce $P$ to its semi-classical Birkhoff normal form (BNF). Namely, there exists a microlocal unitary $h$-Fourier integral operator $U$ associated with $\kappa$, and a classical symbol
$$f(\tau;h)=f_{0}(\tau)+h\,f_{1}(\tau)+h^{2}\,f_{2}(\tau)+\cdots,$$
such that
$$U^{*} P U=f(hD_{t};h),$$
which is an $h$-PDO with constant coefficients, where $hD_{t}$ quantizes $\tau$. We do not provide explicit expressions for $f_{j}(\tau), j\geq 1$, in terms of $p_{j}$, but note that $f_{j}$ depends linearly on $p_{0}, p_{1}, \cdots, p_{j}$ and their derivatives. Of course, BNF allows to get rid of focal points. The section $t=0$ in $\{f_{0}(\tau)=E\}=f_{0}^{-1}(E)$ (the Poincaré section) reduces to a single point $\Sigma=\{a_{E}\}$. Recall from \big(\cite{HL1MR1}, \cite{HL2MR2}\big) that Poisson operator $\mathcal{K}(t,E;h)$ here solves (globally near $\gamma_{0}$), the eigenvalue equation
\begin{equation}\label{weyL1}
\big(f(hD_{t};h)-E\big)\,\mathcal{K}(t,E;h)=0,
\end{equation}
and is given in the special 1-D case by the multiplication operator on $L^{2}(\Sigma)\approx \mathbb{C}^{2}$:
\begin{equation*}
\mathcal{K}(t,E)=e^{\frac{i}{h}\,S(t,E)}\,a(t,E;h),
\end{equation*}
where $S(t,E)$ satisfies the eikonal equation $f_{0}\big(\partial_{t}S(t,E)\big)=E, \,S(0,E)=0$, i.e. $S(t,E)=f_{0}^{-1}(E)\,t=\tau\,t$.

The action of the Schrödinger operator on the oscillatory function $x\mapsto a(x;h)\,e^{\frac{i}{h}\,S(x)}$ is given by
\begin{equation}\label{AActionSCHROScii}
e^{\frac{i}{h}\,S(x)}\,\big(P-E\big)\big(a(x;h)\,e^{\frac{i}{h}\,S(x)}\big)=
\Big(p\big(x,S'(x);h\big)-E\Big)\,a(x;h)-i\,h\,\Big(\beta(x;h)\,\partial_{x}a(x;h)+\frac{1}{2}\,\partial_{x}\beta(x;h)\,a(x;h)\Big)-
h^{2}\,\frac{\partial^{2}a(x;h)}{\partial x^{2}},
\end{equation}
where
$$\beta(x;h)=(\frac{\partial p}{\partial \xi})\big(x,S'(x)\big)$$
In action–angle coordinates this reads
\begin{equation}\label{AActionSCHROSciangle}
e^{-\frac{i}{h}\,S}\,\big(f(hD_{t};h)-E\big)(a\,e^{\frac{i}{h}\,S})=\big(f(\partial_{t} S;h)-E\big)\,a
-i\,h\,\big(\frac{\partial f}{\partial \tau}\,\frac{\partial a}{\partial t}+\frac{1}{2}\,\frac{\partial^{2} f}{\partial \tau^{2}}\,\frac{\partial^{2}S}{\partial t^{2}}\,a\big)-\frac{h^{2}}{2}\,\frac{\partial^{2} f}{\partial \tau^{2}}\,\frac{\partial^{2}a}{\partial t^{2}}
\end{equation}
The phase function $S(t,E)$ is real, smooth, and must satisfy the eikonal equation $f_{0}\big(\partial_{t}S(t,E)\big)=E$, $S(0,E)=0$, i.e. $S(t,E)=f_{0}^{-1}(E)\,t=\tau\,t$. Moreover, $a(t,E;h)=a_{0}(t,E)+h\,a_{1}(t,E)+\cdots$ solves transport equations to any orders in $h$. It then follows that
\begin{equation}\label{AActionSCHROSci01angle}
e^{-\frac{i}{h}\tau t}\,\big(f(hD_{t};h)-E\big)\big(a(t,E;h)\,e^{-\frac{i}{h} \tau t}\big)=
\big(f(\tau;h)-E\big)\,a(t,E;h)-i\,h\,\frac{\partial f(\tau;h)}{\partial \tau}\,\frac{\partial a(t,E;h)}{\partial t}-\frac{h^{2}}{2}\,\frac{\partial^{2} f(\tau;h)}{\partial \tau^{2}}\,\frac{\partial^{2}a(t,E;h)}{\partial t^{2}}
\end{equation}
If the eikonal equation is satisfied, by eliminating the $h$ term in (\ref{AActionSCHROSci01angle}), we obtain the first transport equation \begin{equation}\label{44TRANSPoRT1}
i\,f_{1}(\tau)\,a_{0}(t,E)+\,f'_{0}(\tau)\,\frac{\partial a_{0}(t,E)}{\partial t}=0
\end{equation}
whose solutions are the functions
\begin{equation}\label{Soluu0}
a_{0}(t,E)=C_{0}(E)\,e^{-i t f_{1}(\tau)/f'_{0}(\tau)},\quad C_{0}(E)\in \mathbb{R}
\end{equation}
By eliminating the $h^{2}$ term in (\ref{AActionSCHROSci01angle}), we obtain the second transport equation
\begin{equation}\label{TRANSPoRT12}
i\,f_{1}(\tau)\,a_{1}(t,E)+\,f'_{0}(\tau)\,\frac{\partial a_{1}(t,E)}{\partial t}=
-f'_{1}(\tau)\,\frac{\partial a_{0}(t,E)}{\partial t}-i\,f_{2}(\tau)\,
a_{0}(t,E)+\frac{i}{2}\,f''_{0}(\tau)\,\frac{\partial^{2} a_{0}(t,E)}{\partial t^{2}}
\end{equation}
The homogeneous equation associated with (\ref{TRANSPoRT12}) is the same as (\ref{44TRANSPoRT1}), whose solutions are the functions
$$t\mapsto C_{1}(E)\,e^{-i\,t\,f_{1}(\tau)/f'_{0}(\tau)},\quad C_{1}(E)\in \mathbb{R}$$
Thus, we seek a particular solution of (\ref{TRANSPoRT12}) of the form $t\mapsto D_{1}(t,E)\,e^{-i t f_{1}(\tau)/f'_{0}(\tau)}$. Using the method of variation of constants, we find:
\begin{equation}\label{D1tE}
D_{1}(t,E)=i\,C_{0}(E)\,\widetilde{S}_{2}(E)\,t,\quad \widetilde{S}_{2}(E)=\Big[\frac{1}{2}\,\big(\frac{f_{1}^{2}}{f'_{0}}\big)'(\tau)-f_{2}(\tau)\Big]/f'_{0}(\tau)
\end{equation}
which we normalize by setting $D_{1}(0,E)=0$, so that the general solution to (\ref{TRANSPoRT12}) is
\begin{equation}\label{a1peTITD1tE}
a_{1}(t,E)=\big(C_{1}(E)+i\,C_{0}(E)\,\widetilde{S}_{2}(E)\,t\big)\,e^{-i t f_{1}(\tau)/f'_{0}(\tau)}
\end{equation}
It follows that
\begin{equation}\label{aSOLpeTITD1tE}
a(t,E;h)=\big(C_{0}(E)+h\,C_{1}(E)+i\,h\,C_{0}(E)\,\widetilde{S}_{2}(E)\,t\big)\,e^{i t \widetilde{S}_{1}(E)},\quad \widetilde{S}_{1}(E)=-f_{1}(\tau)/f'_{0}(\tau)
\end{equation}
and therefore
\begin{equation}\label{KTE}
\mathcal{K}(t,E)=e^{it\big(\tau+h \widetilde{S}_{1}(E)\big)/h}\,\big(C_{0}(E)+h\,C_{1}(E)+i\,h\,C_{0}(E)\,\widetilde{S}_{2}(E)\,t\big)
\end{equation}
We define the adjoint $\mathcal{K}^{*}(t,E;h)$ of $\mathcal{K}(t,E;h)$ by
$$\mathcal{K}^{*}(t,E)=a^{*}(t,E;h)\,e^{-\frac{i}{h}\,S(t,E)}=\overline{a(t,E;h)}\,e^{-\frac{i}{h}\,S(t,E)}$$
and
$$\mathcal{K}^{*}(E)\mathbf{\cdot}=\int\mathcal{K}^{*}(t,E)\mathbf{\cdot}\,dt$$
The "flux norme" on $\mathbb{C}^{2}$, is defined by:
\begin{equation}\label{OKkpK01}
(u|v)_{\chi}=\Big(\frac{i}{h}\,\big[f(hD_{t};h), \chi(t)\big]\,\mathcal{K}(t,E)u|\mathcal{K}(t,E)v\Big)
\end{equation}
with the scalar product of $L^{2}(\mathbb{R}_{t})$ on the RHS, and $\chi \in C^{\infty}(\mathbb{R})$ is a smooth step-function, equal to 0 for $t\leq 0$ an 1 for $t\geq 2\pi$. To normalize the solution $\mathcal{K}(t,E)$, we start from:
$$\mathcal{K}^{*}(E)\,\frac{i}{h}\,[f(hD_{t};h), \chi(t)]\,\mathcal{K}(t,E)=\mathrm{Id}_{L^{2}(\mathbb{R})}$$
It is known that
\begin{equation}\label{OKjjK03}
\frac{i}{h}\,[f(hD_{t};h),\chi(t)]=-i\,h\,\chi''(t)+2\,\chi'(t)\,hD_{t}
\end{equation}
A fairly straightforward calculation shows that
\begin{equation}\label{}
I(t,E):=\frac{i}{h}\,\big[f(hD_{t};h),\chi(t)\big]\,\mathcal{K}(t,E)
=\Big(-i\,h\,\chi''(t)\,a(t,E;h)+2\,\chi'(t)\,\tau\,a(t,E;h)-2\,i\,h\,\chi'(t)\,\frac{\partial a(t,E;h)}{\partial t}\Big)\,e^{\frac{i}{h}\,S(t,E)}
\end{equation}
Hence
\begin{equation*}
(u|v)_{\chi}=\big(I(t,E)\,u|\mathcal{K}(t,E)\,v\big)=u\,\overline{v}\,\int_{0}^{2\,\pi}I(t,E)\,\overline{\mathcal{K}(t,E)}\,dt
=u\,\overline{v}\,\Big[2\,\tau\,C^{2}_{0}(E)+h\,\big(4\,C_{0}(E)\,C_{1}(E)\,\tau+2\,C^{2}_{0}(E)\,\widetilde{S_{1}}(E)\big)+\mathcal{O}(h^{2})\Big]
\end{equation*}
We should therefore have:
$$2\,\tau\,C_0^2(E)=1$$
If we choose $C_0>0$, we will have
\begin{equation}\label{OKjjK0hq4}
C_{0}(E)=\big(2\,\tau\big)^{-1/2},\quad C_{1}(E)=\frac{f_{1}(\tau)}{(2\tau)^{3/2}f'_{0}(\tau)}.
\end{equation}
For the found values of $C_0(E), C_1(E)$, we indeed have $(u|v)_\chi=u\,\overline{v}$, and we thus obtain a $\mathcal{K}(t,E)$ normalized to all orders in $h$.

We define $\mathcal{K}_{0}(t,E)=\mathcal{K}(t,E)$ (Poisson operator at $t=0$), and $\mathcal{K}_{2\,\pi}(t,E)=\mathcal{K}_{0}(t-2\,\pi,E)$ (Poisson operator at $t=2\,\pi$). $E$ is an eigenvalue of the operator $f(hD_{t};h)$ if and only if 1 is an eigenvalue of the Monodromy operator
$$M(E)=\mathcal{K}^{*}_{2\,\pi}(E)\,I(t,E)=\int\mathcal{K}^{*}_{0}(t-2\,\pi,E)\,I(t,E)\,dt$$
Note that in dimension 1, the Monodromy operator $M(E)$ reduces to a multiplication operator. It is then shown that
\begin{equation}\label{}
\begin{aligned}
M(E)&=\exp\Big[\frac{i}{h}\,\big(2\,\pi\,\tau+2\,\pi\,h\,\widetilde{S}_{1}(E)\big)\Big]\,\int_{0}^{2\pi}\big[\chi'(t)+2i\pi h \widetilde{S}_{2}(E)\chi'(t)-ihC_{0}^{2}\chi''(t)+\mathcal{O}(h^{2})\big]\,dt\ \\
    &=\exp\Big[\frac{i}{h}\,\big(2\,\pi\,\tau+2\,\pi\,h\,\widetilde{S}_{1}(E)+2\,\pi\,h^{2}\,\widetilde{S}_{2}(E)\big)\Big]\,\big(1+\mathcal{O}(h^{2})\big)
\end{aligned}
\end{equation}
The Bohr-Sommerfeld quantization rule is written as:
\begin{equation}\label{}
f_{0}^{-1}(E)+h\,\widetilde{S_{1}}(E)+h^{2}\,\widetilde{S_{2}}(E)+\mathcal{O}(h^{3})=n\,h,\quad n\in \mathbb{Z}
\end{equation}
Let $S_{1}(E)=2\,\pi\,\widetilde{S}_{1}(E)$, $S_{2}(E)=2\,\pi\,\widetilde{S}_{2}(E)$. Then, since $f_{0}^{-1}(E)=\tau(E)=\displaystyle \frac{1}{2\,\pi}\,\oint_{\gamma_{E}}\xi(x)\,dx=\displaystyle \frac{1}{2\,\pi}\,S_{0}(E)$, and moreover we know that $S_{3}(E)=0$, we obtain:
\begin{equation}\label{BSACANG}
S_{0}(E)+h\,S_{1}(E)+h^{2}\,S_{2}(E)+\mathcal{O}(h^{4})=2\,\pi\,n\,h,\quad n\in \mathbb{Z}
\end{equation}
with
\begin{equation}\label{}
S_{1}(E)=-2\pi\,\big(\frac{f_{1}}{f'_{0}}\big)(\tau),\quad S_{2}(E)=2\pi\,\Big[\frac{1}{2}\,\big(\frac{f_{1}^{2}}{f'_{0}}\big)'(\tau)-f_{2}(\tau)\Big]/f'_{0}(\tau)
\end{equation}
Using equation (\ref{QRFP11}), we have for the Schrödinger operator:
$$f_{1}(\tau)=\frac{1}{2}\,f'_{0}(\tau),\quad f_{2}(\tau)=\frac{1}{8}\,f''_{0}(\tau)+\frac{f'_{0}(\tau)}{24\pi}\,\frac{d}{dE}\int_{\gamma_{E}}V''\big(x(t)\big)\,dt$$
\begin{rem}
For the harmonic oscillator, we have: $V(x)=\displaystyle\frac{1}{2}\,x^{2}$, so $V''(x)=1$ and the derivative with respect to $E$ of the integral is zero, and consequently $f_{2}(\tau)=0$. Then
\begin{equation}\label{}
f_{0}(\tau)=\tau,\quad f_{1}(\tau)=\frac{1}{2},\quad f_{2}(\tau)=0
\end{equation}
The Bohr-Sommerfeld quantization rule for the harmonic oscillator is
\begin{equation}\label{}
\tau\big(E_{n}(h)\big)=h(n+\frac{1}{2})
\end{equation}
Since $\tau(E)=E$, we obtain
\begin{equation}\label{}
E_{n}(h)=h(n+\frac{1}{2})
\end{equation}
This is a special case where semi-classical quantization is exact (no higher-order correction).
\end{rem}
\nocite{*}
\bibliographystyle{plain}
\bibliography{SCHRODINGER(BS)}

\begin{thebibliography}{10}

\bibitem{Argyres}
P.N. Argyres.
\newblock The {B}ohr-{S}ommerfeld {Q}uantization {R}ule and the {W}eyl
  correspondence.
\newblock {\em Physics}, (2):131--199, 1965.

\bibitem{Baklouti}
H.~Baklouti.
\newblock Asymptotique des largeurs de r\'{e}sonances pour un mod\`{e}le
  d'effet tunnel microlocal.
\newblock {\em Ann.IHP (Phys.Th)}, 68(2):1998, 179--228.

\bibitem{CdV1}
Y.~Colin de~Verdi\`{e}re.
\newblock Bohr-sommerfeld {R}ules to {A}ll orders.
\newblock {\em Ann. H.Poincar\'{e}}, 6:925--936, 2005.

\bibitem{YCVPAR1}
Y.~Colin de~Verdi\`{e}re and B.~Parisse.
\newblock \'{E}quilibre instable en r\'{e}gime semi-classique \textrm{II}:
  Conditions de {B}ohr-{S}ommerfeld.
\newblock {\em Annales de l'Institut Henri Poincar\'{e}, Physique
  Th\'{e}orique}, 61(3):347--367, 1994.

\bibitem{DePh}
E.~Delabaere and F.~Pham.
\newblock Resurgence methods in semi-classical asymptotics.
\newblock {\em Ann. Inst. H.Poincar\'{e}}, 71(1):1--94, 1999.

\bibitem{Duncan}
K.P. Duncan and B.L. Gy\"{o}rffy.
\newblock Semiclassical theory of quasiparticles in superconducting state.
\newblock {\em Annals of. Phys}, 298:273--333, 2002.

\bibitem{Fedoriouk}
M.V. Fedoriouk.
\newblock {\em M\'{e}thodes asymptotiques pour les {\'{E}}quations
  {D}iff\'{e}rentielles {O}rdinaires {L}in\'{e}aires}.
\newblock Asymptotic Analysis. Springer, Mouscou, {M}ir edition, 1987.

\bibitem{Gracia00Saz}
A.~Gracia-Saz.
\newblock The symbol of a function of a pseudo-differential operator.
\newblock {\em Ann. Inst. Fourier}, 55(7):2257--2284, 2005.

\bibitem{HARPE0RIII}
B.~Helffer and J.~Sj\"{o}strand.
\newblock Semi-classical analysis for {H}arper's equation \textrm{III}.
\newblock {\em M\'{e}moire No 39, Soc. Math. de France}, 117(4), 1988.

\bibitem{HR84}
Bernard Helffer and Didier Robert.
\newblock Puits de potentiel généralisés et asymptotique semi-classique.
\newblock {\em Ann. Inst. H. Poincaré Phys. Théor}, 41(3):291--331, 1984.

\bibitem{LARsHORmander}
L.~{H\"{o}rmander}.
\newblock {\em The {A}nalysis of {L}inear {P}artial {D}ifferential {O}perators
  I, III}.
\newblock Springer-Verlag, 1983.

\bibitem{AIfaRrou}
A.~Ifa and M.~Rouleux.
\newblock A remark on {B}ohr-sommerfeld {Q}uantization {R}ules for a
  self-adjoint 1-{D} $h$-{P}seudo-differential operator.
\newblock 2020 (submitted).

\bibitem{Littlejohn0}
R.~Littlejohn.
\newblock Lie {A}lgebraic {A}pproach to {H}igher-{O}rder {T}erms.
\newblock 6 2013.

\bibitem{IfaLouRo}
A.~Ifa~H. Louati and M.~Rouleux.
\newblock Bohr-{S}ommerfeld {Q}uantization {R}ules {R}evisited: the {M}ethod of
  {P}ositive {C}ommutators.
\newblock {\em J. Math. Sci. Univ. Tokyo}, 25(2):1--37, 2018.

\bibitem{HL1MR1}
H.~Louati and M.~Rouleux.
\newblock Semi-classical quantization rules for a periodic orbit of hyperbolic
  type.
\newblock {\em Proceedings "Days of Diffraction": DOI: 10.1109/DD.2016.7756858,
  IEEE}, pages 112--117, 2016.

\bibitem{HL2MR2}
H.~Louati and M.~Rouleux.
\newblock Semi-classical resonances associated with a periodic orbit.
\newblock {\em Math. Notes}, 100(5):724--730, 2016.

\bibitem{IfaM'haRo}
A.~Ifa~N. M'hadbi and M.~Rouleux.
\newblock On generalized {B}ohr-{S}ommerfeld quantization rules for operators
  with {P}{T} symmetry.
\newblock {\em Mathematical Notes}, 99(5):673--683, 2016.

\bibitem{VN00}
San V\~{u} Ng\d{o}c.
\newblock Bohr-sommerfeld conditions for integrable systems with critical
  manifolds of focus-focus type.
\newblock {\em Comm. Pure Appl. Math}, 53(2):143--217, 2000.

\bibitem{Olver}
F.W.J. Olver.
\newblock {\em Introduction to Asymptotics and Special Functions}.
\newblock Academic Press, New York, 1974.

\bibitem{CARGooGR}
M.~Cargo A. Gracia-Saz R. Littlejohn~M. Reinsch and P.~de~Rios.
\newblock Moyal star product approach to the {B}ohr-{S}ommerfeld approximation.
\newblock {\em J.Phys.A:Math and Gen}, 38:1977--2004, 2005.

\bibitem{Rouleux}
M.~Rouleux.
\newblock Tunneling effects for $h$-{P}seudodifferential {O}perators,
  {F}eshbach resonances and the {B}orn-oppenheimer approximation.
\newblock {\em Adv. Part. Diff. Eq Wiley-VCH}, 1999.

\bibitem{JO22025}
J.~Sj\"{o}strand.
\newblock Density of states oscillations for magnetic {S}chr\"{o}dinger
  operators, in: {B}ennewitz (ed.).
\newblock {\em Diff. Eq. Math. Phys}, pages 295--345, 1990.

\bibitem{NoSjZw}
S.~Nonnenmacher~J. Sj\"{o}strand and M.~Zworski.
\newblock From {O}pen {Q}uantum {S}ystems to {O}pen {Q}uantum maps.
\newblock {\em Comm. Math. Phys}, 304:1--48, 2011.

\end{thebibliography}
\addcontentsline{toc}{chapter}{Bibliographie}
\end{document}